\documentclass[ssy,preprint]{imsart}

\usepackage[latin1]{inputenc}
\usepackage{amssymb,amsmath,amsthm,amsfonts,dsfont}
\usepackage{graphicx}
\usepackage{tikz}
\usetikzlibrary{decorations.pathreplacing}
\usetikzlibrary{automata, positioning}
\usepackage{tikz-cd}
\usepackage{hyperref}
\usepackage[nodayofweek]{datetime}
\usepackage{xcolor}

\startlocaldefs
\newdateformat{mydate}{\shortmonthname[\THEMONTH]{ }\twodigit{\THEDAY}, \THEYEAR}

\numberwithin{equation}{section}

\newtheorem{theorem}{Theorem}[section]
\newtheorem{lemma}[theorem]{Lemma}

\newtheorem{corollary}[theorem]{Corollary}
\theoremstyle{definition}
\newtheorem{definition}{Definition}[section]
\newtheorem{assumption}{Assumption}[section]
\newtheorem{remark}{Remark}[section]
\newtheorem{example}{Example}[section]

\endlocaldefs

\begin{document}

\begin{frontmatter}
\title{Delay-optimal policies in partial fork-join systems with redundancy and random slowdowns}

\begin{aug}
\author{\fnms{Martin} \snm{Zubeldia}\ead[label=e1]{m.zubeldia.suarez@tue.nl}}


\affiliation{Eindhoven University of Technology and University of Amsterdam}

\address{Martin Zubeldia\\
Department of Mathematics \& Computer Science\\
Eindhoven University of Technology\\
5600 MB Eindhoven, Netherlands.\\
\printead{e1}}

\end{aug}


\begin{abstract}
  We consider a large distributed service system consisting of $n$ homogeneous servers with infinite capacity FIFO queues. Jobs arrive as a Poisson process of rate $\lambda n/k_n$ (for some positive constant $\lambda$ and integer $k_n$). Each incoming job consists of $k_n$ identical tasks that can be executed in parallel, and that can be encoded into at least $k_n$ ``replicas" of the same size (by introducing redundancy) so that the job is considered to be completed when \emph{any} $k_n$ replicas associated with it finish their service. Moreover, we assume that servers can experience random slowdowns in their processing rate so that the service time of a replica is the product of its size and a random slowdown.

  First, we assume that the server slowdowns are shifted exponential and independent of the replica sizes. In this setting we show that the delay of a typical job is asymptotically minimized (as $n\to\infty$) when the number of replicas per task is a constant that only depends on the arrival rate $\lambda$, and on the expected slowdown of servers.

  Second, we introduce a new model for the server slowdowns in which larger tasks experience less variable slowdowns than smaller tasks. In this setting we show that, under the class of policies where all replicas start their service at the same time, the delay of a typical job is asymptotically minimized (as $n\to\infty$) when the number of replicas per task is made to depend on the actual size of the tasks being replicated, with smaller tasks being replicated more than larger tasks.
\end{abstract}

%

\end{frontmatter}

\setcounter{tocdepth}{2}
\tableofcontents

\section{Introduction}

Consider a distributed service system consisting of a large number of servers operating in parallel, where each server can experience random slowdowns in its processing rate. Each incoming job consists of $k$ identical tasks that can be executed in parallel, and that can be encoded into at least $k$ ``replicas" of the same size (by introducing redundancy) so that the job is considered to be completed when \emph{any} $k$ replicas associated with it finish their service. For the case of $k=1$, this corresponds to simply creating copies of the job, dispatching them to different servers, and waiting for the first one to finish. For the general case with $k>1$, our model is a generalization of the classic fork-join model (which assumes $k=n$), and it is motivated by the following applications:
\begin{itemize}
  \item [(i)] {\bf Data file retrieval with redundancy and coding:} Consider a setting where a user wants to retrieve a large file, which can be downloaded from a large set of servers. In order to shorten the download time, the following scheme can be used \cite{joshiCoding}. The file is split into $k$ pieces of equal size, and then the $k$ pieces are encoded into $r\geq k$ pieces of the same size, in a way that allows the original file to be recovered from any $k$ pieces. The user starts to download the $r$ encoded pieces from $r$ different servers, waits for the first $k$ to finish, and then cancels the other downloads. In this setting, the slowdowns in the download come from the congestion in the route between the user that requests the file and the servers from which the pieces are downloaded, as well as from congestion at the server themselves. This encoding scheme was shown to be better than replicating each of the pieces in \cite{MDSbetter}.
  \item [(ii)] {\bf Approximate distributed computing:} Consider a setting where a user wants to compute $k$ gradient estimates in a Parallelized Stochastic Gradient Descent method, executed by a large server farm with $n\gg k$ servers \cite{distributedGradient}. In order to reduce processing time, we attempt to compute $r\geq k$ gradient estimates and we stop when any $k$ computations are finished. In this setting, the slowdowns come from reductions in server processing power due to background processes or similar issues.
\end{itemize}
In general, there are many sources of randomness for the service time of a task. These include:
\begin{itemize}
  \item [(i)] The intrinsic task size variability.
  \item [(ii)] Slowdowns in the local processing rate of the server due to exogenous interferences (such as background processes).
  \item [(iii)] Network congestion that interferes with communication.
\end{itemize}
The intrinsic task size variability is always a source of randomness in the service times (even when there are no slowdowns), and it is captured in almost every model in the literature, including the two models that we will work with. However, not all other sources of randomness are relevant for all applications. For instance, network congestion interfering with communication can create significant slowdowns for the retrieval of data from a server, whereas the reduction in the local processing power of a server can barely affect it. This behavior is more accurately reflected by the independent exponential slowdowns model introduced in Section \ref{sec:exponential}. On the other hand, random reductions in the processing power of local servers have a greater impact on the slowdowns in distributed computing than network congestion. This behavior is more accurately reflected by the size-based slowdowns model introduced in Section \ref{sec:sizeBased}.\\

Our objective is to understand the best possible performance of such systems and to propose near-optimal policies, with emphasis on the asymptotic regime when $n$ is large. In particular, our performance metric is the {\bf delay} of a typical job, i.e., the time between a job arrives to the system, and the time when $k$ replicas associated with it finish their services. The delay can be decomposed as the sum of the {\bf service time}, i.e., the total amount of time that at least one replica associated with the job is in service, and the {\bf waiting time}, i.e., the total amount of time that no replica associated with the job is in service.\\

A control policy for such systems must specify how many redundant tasks (called {\bf replicas} from now on) to create, when to create them (their creation can be staggered in time), and to which servers to send them. Furthermore, replicas can also be prematurely cancelled. All of these decisions have the potential of reducing the delay of a typical job.
\begin{enumerate}
\item For example, suppose that we send $r$ replicas to different servers, but as soon as any $k$ replicas associated with the same job start their service, all other replicas associated with the same job are cancelled. In this case, the $k$ replicas in service correspond to the $k$-th shortest queues out of the original $r$. This effectively reduces the queueing delay of replicas, and thus the queueing delay of the job. Furthermore, since the other replicas never start service, they do not consume any resources and do not affect the delays of other jobs.
\item Even if all replicas experience zero queueing delays, the creation of more replicas can reduce the delay of the job. This is because replicas processed by servers with higher processing rates (i.e., that experience less severe slowdowns) finish their service earlier. Thus, if $r>k$ replicas start their service at the same time, the delay of the job will be the $k$-th smallest out of the $r$ service times of the replicas associated with it.
\item While the creation of more replicas always reduces the delay of an individual job, too many replicas can lead to instability. Furthermore, even before the system becomes unstable, the cost of the increased congestion may overshadow the reduction in service times.
\end{enumerate}
The combination of having multiple tasks per job, the ability to create replicas, and the random slowdowns in the service rates, make replication-based systems difficult to analyze, even under simple control policies.

\subsection{Previous work}
Most of the prior theoretical work on parallel service systems with redundancy has been made for the case of a single task per job (i.e., the case $k_n=1$), and under the (unrealistic) \emph{independent runtimes model}, which stipulates that the times that replicas require to complete their service are i.i.d., regardless of whether they are associated with the same job or not. Using this model, a body of work has focused on characterizing optimal replication policies based on the log-concavity of the complementary cdf of the service times \cite{joshiConcavity,Koole08,generalists,optimalShroff}. In particular, it has been shown that if the complementary cdf of the service time distribution of a replica is log-convex, then the minimum delay is achieved when replicas are sent to all the servers. Furthermore, in \cite{redundancyD,lowLatency,joshiTradeoff,joshiTradeoff2} the authors explore the tradeoffs between the delays and the resources utilized, for several classes of service time distributions (mainly discrete and log-convex). The stability and performance for difference scheduling policies at the queues was studied in \cite{stabilityRedundancy,gardner2019}.

On the other hand, there is some recent work \cite{decoupledSlowdown} that introduces a more realistic model where service times of replicas associated with the same job are correlated. Policies under this model are significantly harder to analyze, and the available theoretical results are limited. In \cite{poloczek16,joshiCoding}, the authors obtain results about the performance of particular policies. Furthermore, in \cite{decoupledFairness} the authors develop policies that result in fair delays for multiple classes of jobs, under replication constraints.

Results from different settings can also cast a light on delay-optimal replication schemes. For example, in \cite{semReplication} the authors analyse a discrete-time parallel service system, where the service times of replicas are i.i.d. and geometrically distributed, and where replicas can be created and deleted after each time slot, so that the number of servers working on a job can change from slot to slot. In this setting, it has been shown that the delay is minimized when all servers are used all of the time, and the number of replicas associated with each job are all equal (or differ by at most 1). Furthermore, in \cite{decoupledOptimality} the authors analyse a different parallel service system where, instead of replication, jobs are amenable to parallel processing, with a sublinear improvement in processing rate. For the case of exponentially distributed jobs it was shown that if the number of servers processing a job can be changed at any point in time, then the delay is minimized when all servers are used all of the time, and the number of servers that process each job are all equal (or differ by at most 1).

Finally, for the more general case where $k_n>1$, our model is a generalization of the classic fork-join model, where the number of tasks per job is equal to the number of servers (i.e., where $k_n=n$). In this setting, tight characterizations of the delay are only known for the two-server case (see \cite{classicForkJoin} for a detailed survey). Although there are no tight delay characterizations for fork-join models, there are several asymptotic and non-asymptotic bounds \cite{kumarEtAl,leeEtAl,nelsonEtAl,rizkEtAl,shahEtAl,multiTaskAsymptotics}, with different levels of tightness.

\subsection{Our contribution}
We consider a broad family of control policies, which includes most policies considered in the literature, and work towards characterizing the achievable delay performance of jobs under two different models for the slowdowns.
\begin{enumerate}
\item We consider the $S\& X$ model, first introduced in \cite{decoupledSlowdown}, which assumes that the slowdowns are independent from the task sizes. This is a plausible model of the slowdowns which fluctuate on a time scale that is slower than the typical delay. Moreover, we assume that the slowdowns are exponential, which are observed, for example, in the download times from Amazon servers \cite{empirical1,empirical2}. Under this model, our first contribution is a universal lower bound for the expected delay of a typical job under any control policy, which provides a benchmark for any practical policy. Surprisingly, this lower bound is independent of the task size distribution, of the distribution of the inter-arrival times, and of the number of servers. Our second contribution is the introduction of simple control policies that asymptotically achieve the lower bound, under minor technical conditions. For these asymptotically optimal policies, the number of replicas created per task (i.e., the quantity $r/k$) is independent of the task sizes, and of the number of tasks.
\item We also consider a new Size-based slowdown model, under which the distribution of the slowdowns is a function of the task sizes. This reflects the fact that longer replicas should ``average out" the slowdowns and experience less variability in their service times than short job. This is more accurate for modeling slowdowns that are in the same time-scale as the delays (e.g., the slowdowns in distributed computing). Under this model, we consider policies where all replicas associated with the same job start their service at the same time. For policies that are optimal under this restriction, we show that the number of replicas per task is nonincreasing in the size of the task. This is consistent with current practice, but to the best of our knowledge this is the first tractable model that justifies this practice.
\end{enumerate}

\subsection{Outline of the paper}
The rest of the paper is organized as follows. In Section \ref{sec:replication_model_assumption} we introduce the general modeling assumptions, and the policies that are considered throughout the paper. In Section \ref{sec:exponential} we introduce the first model for the slowdowns, and the main results for that model. In Section \ref{sec:sizeBased} we introduce a new sized-based model for the slowdowns, and the main results for this new model. Finally, in Section \ref{sec:conclusions} we present our conclusions and suggestions for future work.

\section{Modeling assumptions and performance metrics} \label{sec:replication_model_assumption}
Throughout this paper we consider a system consisting of $n$ homogeneous servers, which can experience random slowdowns in their processing rates, and where each server is associated with an infinite capacity FIFO queue. We assume that jobs arrive to the system as a Poisson process of rate $\lambda n/k_n$, for some fixed $\lambda>0$ and some positive integer $k_n\leq n$. Each job consists of $k_n$ tasks of the same (albeit random) size. Task sizes are i.i.d. across different jobs, and have unit mean. Furthermore, we assume that we can encode the $k_n$ tasks into any number of at least $k_n$ replicas (of the same size as the tasks) such that a job is finished when \emph{any} $k_n$ replicas associated with it finish their service (at which point all remaining replicas associated with the same job are immediately removed from the system).\\

\noindent{\bf Service time variability:}
Let $X_j$ be the size of the replicas associated with the $j$-th job. As mentioned earlier, the $\{X_j\}_{j\geq 1}$ are i.i.d. The number of replicas created and associated with each job can be random, as it depends on the control policy, congestion, and other factors. Let $S_{j,r}$ be the slowdown that would be experienced by the $r$-th replica associated with the $j$-th job, if such a replica were to start its service. In that case, we assume that the time required for the $r$-th replica associated with the $j$-th job to finish its service is equal to $X_j(1+S_{j,r})$. Moreover, we assume the following.

\begin{assumption}\label{ass:basicSlowdown}
$ $
\begin{itemize}
   \item [(i)] There exists a family of cumulative distribution functions $\{F_x:x\geq 0\}$ such that, for all $x\geq 0$, and for every $j\geq 1$ and $R\geq 1$, we have
\[ \mathbb{P}\left(\left. \bigcap_{r=1}^{R} \big\{ S_{j,r}\leq s_r \big\} \,\right|\, X_j=x\right) = \prod_{r=1}^{R} F_x(s_r), \]
for all $s_1,\dots,s_{R}\geq 0$.
   \item [(ii)] The random sequences of slowdowns $\big\{(S_{j,r}:r\geq 1)\big\}_{j\geq 1}$ are independent, i.e., for every $J\geq 1$, we have
   \[ \mathbb{P}\left( \bigcap_{j=1}^J \Big\{\big(S_{j,r}:r\geq 1\big) \in A_j \Big\} \right) = \prod_{j=1}^J \mathbb{P}\big( (S_{j,r}:r\geq 1) \in A_j \big), \]
   for all measurable sets $A_1,\dots,A_J$.
\end{itemize}
\end{assumption}

\begin{remark}
  The first part of Assumption \ref{ass:basicSlowdown} asserts that, conditioned on the task size, the slowdowns experienced by different replicas associated with the same job are independent and identically distributed, and that the conditional distribution of the slowdowns is the same across different jobs. Combining this with the fact that replicas' sizes are i.i.d., we get that the random sequences of slowdowns $\big\{(S_{j,r}:r\geq 1)\big\}_{j\geq 1}$ are identically distributed, i.e., that
  \[ (S_{j,r}:r\geq 1) \overset{d}{=}(S_{j',r}:r\geq 1), \]
  for all $j,j'\geq 1$. Further combining this with the second part of Assumption \ref{ass:basicSlowdown}, we conclude that the random sequences of slowdowns $\big\{(S_{j,r}:r\geq 1)\big\}_{j\geq 1}$ are i.i.d.
\end{remark}


In the two models that we consider, there will be different additional assumptions on the family of cumulative distribution functions $\{F_x:x\geq 0\}$.

\subsection{Admissible control policies}\label{sec:assumptions}
In this subsection, we introduce a broad family of control policies for our system. In order to do this, we start by introducing the extended queue state process $\mathcal{Q}_n(\cdot)$. The extended queue state $\mathcal{Q}_n(t)$ at time $t$ is a set, where each element corresponds to a job in the system. In particular,
\[ \Big(x,u,\big\{d_1,\dots,d_k\big\},\big\{(n_1,y_1,e_1),\dots,(n_r,y_r,e_r)\big\}\Big)\in \mathcal{Q}_n(t) \]
if at time $t$ there is a job in the system that satisfies the following:
\begin{itemize}
  \item [(i)] the tasks associated with the job have size $x$,
  \item [(ii)] the job arrived $u$ units of time ago
  \item [(iii)] $k$ replicas associated with the job have already finished their service, $d_1,\dots,d_k$ units of time after the arrival of the job,
  \item [(iv)] the job has a total of $r$ replicas associated with it still in the system, in queues $n_1,\dots,n_r$, which
  \begin{enumerate}
    \item are in positions $y_1,\dots,y_r$ in their respective queues { (with the convention that if a replica is in service then it is in position $0$)},
    \item have elapsed service times equal to $e_1,\dots,e_r$.
  \end{enumerate}
\end{itemize}
{ For example, if at time $t$ there is a job in the system with tasks of size $x$, which arrived $u$ units of time ago, and no replicas were created yet, we have $(x,u,\emptyset,\emptyset)\in \mathcal{Q}_n(t)$.}\\

The evolution of the process $\mathcal{Q}_n(\cdot)$ is partly driven by intrinsic features of our queueing model, such as the arrival process, the size of incoming tasks, the slowdown of the replicas, the FIFO queues, and the fact that each job leaves the system (together with all of the replicas associated with it) at the moment that $k_n$ of its replicas finish their service. On the other hand, the control policy determines:
\begin{itemize}
  \item [(i)] when to create new replicas for jobs in the system,
  \item [(ii)] where to dispatch the newly created replicas,
  \item [(iii)] when to cancel replicas in the system.
\end{itemize}
In general, we consider policies that use the current extended queue state to make these decisions. This is formalized in the following assumption.

\begin{assumption}[Markovianity]\label{ass:markovianity}
  The extended queue state process $\mathcal{Q}_n(\cdot)$ is Markov.
\end{assumption}

Note that this assumption implies that the control policies considered are completely characterized by the Markov kernel of the corresponding extended queue state process $\mathcal{Q}_n(\cdot)$.

Moreover, the extended queue state process $\mathcal{Q}_n(\cdot)$ is rich enough to allow for policies that take into account task size, time since the last replica associated with a job finished its service, and more, in order to decide when to create and cancel replicas.\\

We now introduce an assumption that restricts when replicas can be cancelled.

\begin{assumption}[No late cancellations]\label{ass:noCancelInService}
  Replicas with positive elapsed times cannot be cancelled { (unless $k_n$ of them have finished their service)}.
\end{assumption}

This assumption is introduced to keep the policies tractable, but it can lead to a degradation of performance in some cases. For example, suppose that the expected remaining service time of a replica with positive elapsed time is larger than the expected service time of a new replica. In that case, it is better to cancel the preexisting replica and replace it with a new one.

On the other hand, when the distribution of the slowdowns has a non-decreasing hazard rate, the expected remaining service time of a replica with positive elapsed time is always smaller than the expected service time of a new replica. In that case, it is never a good idea to cancel a replica that is being processed to replace it with a new one. Thus, Assumption \ref{ass:noCancelInService} is not very restrictive when the slowdowns have non-decreasing hazard rates (as will be the case in Section \ref{sec:exponential}).\\

Moreover, Assumption \ref{ass:noCancelInService} implies the following.

\begin{lemma}\label{lem:boundedReplicas}
  Under Assumption \ref{ass:noCancelInService}, at most $n+k_n-1$ replicas associated with the same job are ever in service.
\end{lemma}
\begin{proof}
Suppose that at least $n+k_n$ replicas associated with the same job are eventually in service. Since there can only be at most $n$ replicas in service at the same time (because we have only $n$ FIFO queues), then in order for the ($n+k_n$)-th replica to start its service, at least $k_n$ replicas have to have finished their services before. This is a contradiction, because all replicas associated with the same job leave the system as soon as $k_n$ replicas associated with it finish their services.
\end{proof}

We are now ready to define the admissible control policies that will be used throughout this paper.

\begin{definition}[Admissible policies]
We say that a control policy is \emph{admissible} if the corresponding extended queue state process $\mathcal{Q}_n(\cdot)$ satisfies assumptions \ref{ass:markovianity} and \ref{ass:noCancelInService}.
\end{definition}

\begin{remark}
Note that a decision maker implementing the policies introduced in this section has full knowledge of the state of all queues in the system, including the size of all tasks. However, the decision maker has no information about the slowdowns of the servers, other than what can be inferred from the state of the queues.
\end{remark}

\subsection{Stability and performance metric}
We say that an admissible policy is {\bf stable} if the corresponding process $\mathcal{Q}_n(\cdot)$ admits a unique invariant probability measure. For stable policies, the performance metric of interest is the {\bf expected delay} of a typical job, i.e., the steady-state expectation of the time between the moment a job arrives to the system, and the moment that $k_n$ replicas associated with it finish their service.

The delay of a typical job (denoted by $W_n$) can be decomposed as the sum of two components. The first component is the {\bf service time} (denoted by $W^s_n$), and it is defined as the total amount of time that at least one replica associated with the job is in service. The second component is the {\bf queueing delay} (denoted by $W^q_n$), and it is defined as the total amount of time between the arrival and the departure of the job that no replica associated with it is in service.

\begin{remark}
Note that, since $k_n$ replicas associated with a job need to finish their service before the job leaves the system, and since the replicas can be staggered in time, a job can be in service intermittently over time. In particular, while a job is in the system, it incurs service time when at least one of the replicas associated with it is in service. The rest of the time it is incurring queueing delay. This is akin to the service time and queueing delay of a job in a preemptive queue, where a job can start and resume its service many times.
\end{remark}


\section{Independent exponential slowdowns}\label{sec:exponential}
In this section, we explore the stability and delay performance of admissible policies under the assumption that slowdowns are exponential and independent from the replica size. In particular, in Subsection \ref{sec:expLowerBound} we obtain necessary and sufficient conditions in the parameters of the system for the existence of stable policies, and we obtain a universal lower bound on the expected delay of any admissible policy. Moreover, in Subsection \ref{sec:upperBoundFinite} we introduce a pair of admissible policies and show that they are asymptotically delay-optimal under certain conditions on the moments of the task sizes.\\

Throughout this section, we assume the following.

\begin{assumption}\label{ass:independentExponential}
 For all $j,r\geq 1$, we have
 \[ \mathbb{P}\big(S_{j,r} \leq s\big)=1-e^{-\mu s}, \]
 for all $s\geq 0$.
\end{assumption}

Note that combining Assumption \ref{ass:independentExponential} with Assumption \ref{ass:basicSlowdown} we get that, for all $j\geq 1$, the slowdowns $\{S_{j,r}\}_{r\geq 1}$ are i.i.d. exponential random variables with mean $1/\mu$, which are also independent from the replicas' size $X_j$. This corresponds to the $S\&X$ model introduced in \cite{decoupledSlowdown}, for the special case where slowdowns are exponential. On the practical side, exponential slowdowns are observed, for example, in the download times from Amazon servers \cite{empirical1,empirical2}.\\

In what follows, we present a systematic approach to the analysis and design of policies, which culminates with the introduction and analysis of asymptotically optimal policies.

\subsection{Delay lower bound}\label{sec:expLowerBound}
In this subsection, we first obtain a necessary and sufficient condition on the parameters of the system for the existence of stable admissible policies. Then, we establish a lower bound on the expected delay of a typical job for any stable admissible policy.

\begin{theorem}\label{thm:stability}
  Under Assumption \ref{ass:independentExponential}, there exists a stable admissible policy if and only if
  \begin{equation}\label{eq:condition}
   \lambda < \frac{1}{1+\frac{1}{\mu}} .
  \end{equation}
\end{theorem}
The sufficiency of this condition is established using a very simple policy: each time a new job arrives, $k_n$ replicas are created and dispatched uniformly at random among the $n$ servers. Under this policy, it is easily checked that each queue behaves as a M/G/1 queue with arrival rate $\lambda$ and expected service time equal to $1+1/\mu$, which is known to be stable as long as Equation \eqref{eq:condition} is satisfied.

On the other hand, the necessity of Equation \eqref{eq:condition} is established by showing that increasing the number of replicas created cannot enlarge the stability region. The proof of necessity is given in Appendix \ref{proof:stability}.

\begin{remark}
Note that Equation \eqref{eq:condition} is equivalent to
\[ \frac{1}{\lambda} - \frac{1}{\mu} >1. \]
Thus, the existence of stable admissible policies only depends on the quantity
\begin{equation}\label{eq:r^*}
  r^*= \frac{1}{\lambda} - \frac{1}{\mu}.
\end{equation}
This is a key quantity that will be ubiquitous throughout this paper.
\end{remark}

The condition in Equation \eqref{eq:condition} does not depend on $n$ or $k_n$, but it critically depends on the assumption that the slowdowns are exponential and independent from the task size. Note the above described policy, which creates only $k_n$ replicas, has the largest stability region (in the sense that the policy is stable for the largest possible set of values of $\lambda$). In contrast, as shown in \cite{poloczek16}, the largest stability region can sometimes be augmented by using more than $k_n$ replicas, when the exponential slowdown assumption is relaxed.\\


We now introduce a lower bound for the expected delay of a job, under any stable admissible policy.

\begin{theorem}\label{thm:boundGeneral}
  Fix some $n$. Consider a stable admissible policy, i.e., a policy for which the process $\mathcal{Q}_n(\cdot)$ satisfies assumptions \ref{ass:markovianity} and \ref{ass:noCancelInService}, and under which the process $\mathcal{Q}_n(\cdot)$ admits a unique invariant probability measure $\pi_n$. Then, under Assumption \ref{ass:independentExponential}, we have
  \begin{equation}\label{eq:firstLowerBound}
   \mathbb{E}_{\pi_n}[W_n] \geq 1+ \frac{1}{\mu}\sum\limits_{i=1}^{k_n} \frac{1}{k_n \left(\frac{1}{\lambda}-\frac{1}{\mu}\right)-i+1}.
  \end{equation}
\end{theorem}
The proof is given in Appendix \ref{proof:lowerBound}.\\

A special case of interest is when $k_n\to\infty$ as $n\to\infty$, in which case we have the following result.

\begin{corollary}\label{cor:largeK}
  Under the same assumptions as in Theorem \ref{thm:boundGeneral}, and when $k_n\to\infty$ as $n\to\infty$, we have
  \begin{equation}\label{eq:secondLowerBound}
  \liminf\limits_{n\to\infty} \mathbb{E}_{\pi_n}[W_n] \geq 1+ \frac{1}{\mu}\log\left( \frac{\frac{1}{\lambda}-\frac{1}{\mu}}{\frac{1}{\lambda}-\frac{1}{\mu}-1} \right).
  \end{equation}
\end{corollary}

Surprisingly, the lower bounds in equations \eqref{eq:firstLowerBound} and \eqref{eq:secondLowerBound} are independent from the task size distribution and from the number of servers $n$. Moreover, it can be shown to hold even when the arrivals form a renewal process. However, it does depend on the number of tasks per job $k_n$, and on the quantity
\[ r^*=\frac{1}{\lambda}-\frac{1}{\mu}, \]
which also appeared in Theorem \ref{thm:stability}. Using this notation, we can rewrite the lower bound in Corollary \ref{cor:largeK} as
\[ \liminf\limits_{n\to\infty} \mathbb{E}_{\pi_n}[W_n] \geq 1+ \frac{1}{\mu}\log\left( \frac{r^*}{r^*-1} \right), \]
and the one in Theorem \ref{thm:boundGeneral} as
\[ \mathbb{E}_{\pi_n}[W_n] \geq 1 + \frac{1}{\mu}\sum\limits_{i=1}^{k_n} \frac{1}{k_n r^*-i+1}. \]
Note that, if $k_n r^*$ is an integer, then
\[ \frac{1}{\mu}\sum\limits_{i=1}^{k_n} \frac{1}{k_n r^*-i+1} \]
is the expectation of the $k_n$-th order statistic of $k_n r^*$ i.i.d. exponential random variables with mean $1/\mu$. Thus, if an admissible policy creates $k_n r^*$ replicas per job, and they all start their service at the same time, the expected service time of the jobs would match the lower bound on the expected delay given in Theorem \ref{thm:boundGeneral}. However, the expected delay of a typical job under such a policy might be higher due to queueing delays. In fact, it can be checked that such a policy would result in a system that is critically loaded (i.e., it would have a load of $\rho=1$ using typical queueing theory notation).\\

The above observations seem to indicate that no admissible policy could match the lower bound for any finite $n$. However, even if an admissible policy is not optimal for any finite $n$, it may be asymptotically optimal as $n\to\infty$. Such policies could be obtained by following the following design principles.
\begin{enumerate}
  \item Replicas are sent to idle servers whenever possible.
  \item For each job, the number of replicas associated with it ($r_n$) is such that $\frac{r_n}{k_n} \uparrow r^*$ as $n\to\infty$. This forces a heavy-traffic regime through increased replication.
  \item The forced heavy-traffic is light enough to ensure that new replicas can be sent to idle servers with high probability.
\end{enumerate}
In the following subsection, we formally introduce policies designed according to these principles.

\subsection{Asymptotically optimal policies}\label{sec:upperBoundFinite}
In this subsection, we introduce a pair of admissible policies that are asymptotically optimal under a mild technical condition on the moment of the task sizes, and on the integrality of $r^* k_n$. While both policies are stated and can be implemented for any value of $k_n$, one of them appears to only be tractable for the case of $k_n=1$.

\paragraph{Intuition behind the policies}

While the reasoning at the end of the previous subsection assumed that $r^* k_n$ is an integer, this is not always the case. However, we can always create a random number of replicas for each job so that we create $r^* k_n$ replicas per job in average. In particular, let $p_{k_n}\in(0,1]$ be such that
\[ p_{k_n}\lceil r^* k_n \rceil + (1-p_{k_n})(\lceil r^* k_n \rceil -1)= r^* k_n. \]
Then, creating $\lceil r^* k_n \rceil$ replicas per job with probability $p_{k_n}$, and $\lceil r^* k_n \rceil-1$ replicas per job with probability $1-p_{k_n}$, yields an average of $r^* k_n$ replicas per job.

It can be checked that if all replicas associated with the same job start their service at the same time, a policy that creates the number of replicas described above would require all $n$ servers to process the incoming jobs. As a result, such policies would suffer from the same instability problem as the one mentioned in the previous subsection for the case where $r^* k_n$ is an integer: both subsystems would be critically loaded (i.e., they  would have a load of $\rho=1$ using typical queueing theory notation).

In order to obtain stable policies that have the same asymptotic performance as the ones described above, we force the system to be in heavy-traffic by slightly reducing the average number of replicas created per job (while still creating the same number of replicas per job in the limit as $n\to\infty$), as follows.

Let $\alpha\in(1/2,1)$ be a constant. Suppose that $\lceil r^* k_n \rceil$ replicas per job are created with probability $\left( p_{k_n} - 2n^{\alpha-1} \right)^+$, and $\lceil r^* k_n \rceil-1$ replicas per job are created with probability $1-\left( p_{k_n} - 2n^{\alpha-1} \right)^+$. For all $n$ large enough, this yields an average of $r^* k_n - 2n^{\alpha-1}$ replicas per job, which is slightly less than the $r^* k_n$ replicas per job that would critically load the system. It can be checked that if all replicas associated with the same job start their service at the same time, a policy that creates the number of replicas described above would need
\begin{equation}\label{eq:busy1}
 \frac{\lambda n}{k_n} \left( p_{k_n} - 2n^{\alpha-1} \right)^+\left(\lceil r^* k_n \rceil+ \frac{k_n}{\mu} \right)
\end{equation}
servers to handle the jobs that start with $\lceil r^* k_n \rceil$ associated replicas, and
\begin{equation}\label{eq:busy2}
 \frac{\lambda n}{k_n}\left[1- \left( p_{k_n}-2n^{\alpha-1} \right)^+ \right]\left( \lceil r^* k_n \rceil-1+ \frac{k_n}{\mu} \right)
\end{equation}
servers to handle the jobs that start with $\lceil r^* k_n \rceil-1$ associated replicas. Combining these two quantities with the definitions of $p_{k_n}$ and $r^*$, it leads to a total of
\[ n-\frac{2\lambda n^\alpha}{k_n} \]
servers. Thus, there are
\begin{equation}\label{eq:idle}
\frac{2\lambda n^\alpha}{k_n}
\end{equation}
``spare'' servers in the system.

\paragraph{Description of the policies}
We now introduce the two policies. In order to simplify the analysis, for both policies we partition the $n$-server system into two subsystems, one with
\[ n^{(1)}\triangleq \left\lfloor \frac{\lambda n}{k_n} \left( p_{k_n} - 2n^{\alpha-1} \right)^+\left(\lceil r^* k_n \rceil+ \frac{k_n}{\mu} \right) + \frac{\lambda n^{\alpha}}{k_n} \right\rfloor \]
servers, and the other one with
\[ n^{(2)}\triangleq \left\lceil \frac{\lambda n}{k_n}\left[1- \left( p_{k_n}-2n^{\alpha-1} \right)^+ \right]\left( \lceil r^* k_n \rceil-1+ \frac{k_n}{\mu} \right) + \frac{\lambda n^{\alpha}}{k_n} \right\rceil \]
servers. {  Note that the number of servers in the first (respectively, second) subsystem is equal to the number of servers required to process the jobs that start with $\lceil r^* k_n \rceil$ (respectively, $\lceil r^* k_n \rceil-1$) replicas associated with it, given in Equation \eqref{eq:busy1} (respectively, Equation \eqref{eq:busy2}), plus half of the ``spare'' servers given in Equation \eqref{eq:idle}.} Without loss of generality, we assume that $n^{(1)}$ is a multiple of $\lceil r^* \rceil$, and that $n^{(2)}$ is a multiple of $\lceil r^* \rceil-1$. If this is not the case, we can always choose not to use some of the original $n$ servers. When a job arrives to the system, it is sent to the first subsystem with probability $\big(p_{k_n}-2n^{\alpha-1}\big)^+$, and to the second subsystem with probability $1-\big(p_{k_n}-2n^{\alpha-1}\big)^+$. Then, we do one of the following.
\begin{itemize}
\item {\bf Full Replication with Early Cancellation (FREC)}: If at the time of an arrival to the first (second) subsystem there are at least $\lceil r^* k_n \rceil$ (respectively, $\lceil r^* k_n \rceil -1$) idle servers, then $\lceil r^* k_n \rceil$ (respectively, $\lceil r^* k_n \rceil -1$) replicas are created and dispatched to idle servers. Otherwise, replicas are created and dispatched to \emph{all} servers in the subsystem. When $\lceil r^* k_n \rceil$ (respectively, $\lceil r^*k_n \rceil -1$) of these replicas have started their service, all other replicas associated with the same job are cancelled and immediately leave the system. (Note that this is compatible with Assumption \ref{ass:noCancelInService}, as canceled replicas have not started their service.)
\item {\bf Dummy Queues (DQ)}: If at the time of an arrival to the first (second) subsystem there are at least $\lceil r^* k_n \rceil$ (respectively, $\lceil r^* k_n \rceil -1$) idle servers among the first $n^{(1)}-\lambda n^\alpha/2k_n$ (respectively, $n^{(2)}-\lambda n^\alpha/2k_n$) servers in the subsystem, then $\lceil r^* k_n \rceil$ (respectively, $\lceil r^* k_n \rceil -1$) replicas are dispatched to those idle servers. Otherwise, $k_n$ replicas are dispatched uniformly at random among the last $\lambda n^\alpha/2k_n$ servers of the first (respectively, second) subsystem.
\end{itemize}
Note that when there are enough idle servers, both of these policies send all replicas to idle servers. Thus, if the systems are in light enough heavy-traffic, the vast majority of replicas will be sent to idle servers in both cases. We also note that although the FREC policy seems superior, it appears to be tractable only for the case of $k_n=1$.

\paragraph{Main results}

We now present the main results of this subsection.

\begin{theorem}\label{thm:UpperBound1}
 Suppose that $k_n=1$, for all $n$. If
 \[ \lambda < \frac{1}{1+\frac{1}{\mu}}, \]
 then the FREC policy is stable, and
  \begin{equation*}\label{eq:upper1}
   \lim_{n\to\infty} \mathbb{E}[W^s_n]=1+\frac{p_1}{\mu\lceil r^* \rceil} + \frac{1-p_1}{\mu(\lceil r^* \rceil-1)}.
  \end{equation*}
  Furthermore, if $\mathbb{E}\big[X^{2+\epsilon}\big]<\infty$ for some $\epsilon>0$, then
  \[ \lim_{n\to\infty} \mathbb{E}[W^q_n]= 0. \]
\end{theorem}
The proof is given in Appendix \ref{proof:upperBound}.

\begin{theorem}\label{thm:UpperBoundK}
   If
   \[ \lambda < \frac{1}{1+\frac{1}{\mu}}, \]
   then the DQ policy is stable for all $n$ large enough, and:
  \begin{itemize}
    \item [(i)] If $k_n=k$ for all $n$, then
   \begin{equation}\label{eq:upperK}
   \lim_{n\to\infty} \mathbb{E}[W_n^s] = 1+ \frac{1}{\mu}\sum\limits_{i=1}^{k} \left( \frac{p_k}{\lceil r^* k \rceil -i+1} + \frac{1-p_k}{\lceil r^* k \rceil -i} \right).
   \end{equation}
   Furthermore, if $\mathbb{E}\big[X^{2+\epsilon}\big]<\infty$ for some $\epsilon>0$, then
  \[ \lim\limits_{n\to\infty} \mathbb{P}(W_n^q>0) =0. \]
  \item [(ii)] If $k_n\to\infty$ as $n\to\infty$, then
   \begin{equation}\label{eq:upperInf}
   \lim_{n\to\infty} \mathbb{E}[W_n^s] = 1+ \frac{1}{\mu}\log\left( \frac{r^*}{r^*-1} \right).
   \end{equation}
   Furthermore, if $\mathbb{E}\big[X^{2+\epsilon}\big]<\infty$ for some $\epsilon>0$, $k_n\in O\big(n^\beta\big)$ for some $\beta<1/5$, and $\alpha>\max\big\{(2+4\beta)/3,\,(1+5\beta)/2\big\}$, then
  \[ \lim\limits_{n\to\infty} \mathbb{P}(W_n^q>0) =0. \]
  \end{itemize}

\end{theorem}
The proof is given in Appendix \ref{proof:upperBound2}.

\begin{remark}
Recall that if $k_n=k$ for all $n$, then Theorem \ref{thm:boundGeneral} gives us the lower bound
\begin{equation}\label{eq:lowerK}
 \mathbb{E}[W^s_n]\geq 1+ \frac{1}{\mu}\sum\limits_{i=1}^{k} \frac{1}{ r^* k -i+1}.
\end{equation}
In general, there is a small gap between this lower bound and the asymptotic service time of the DQ policy (Equation \eqref{eq:upperK}). However, when $r^* k$ is an integer, there is no gap. This is depicted in Figure \ref{fig:upperVsLower}, where we plot the expressions in equations \eqref{eq:lowerK} and \eqref{eq:upperK} as functions of { $r^*=1/\lambda-1/\mu$ when we keep $\mu$ and $k$ fixed and we vary $\lambda$.}

\begin{center}
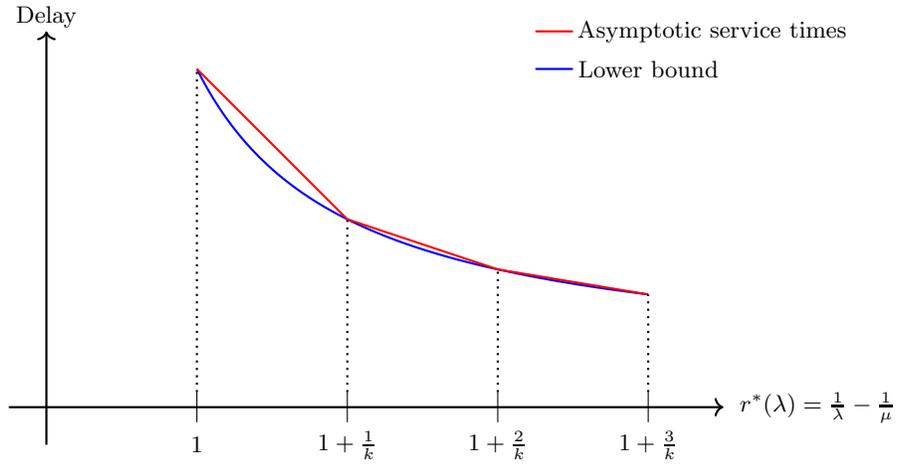
\begin{figure}[ht!]
  \begin{tikzpicture}[scale=1]
    \draw [->, thick] (-0.5,0) -- (9,0);
    \draw (10.25,0) node {$r^*(\lambda)=\frac{1}{\lambda}-\frac{1}{\mu}$};
    \draw [->, thick] (0,-0.5) -- (0,5);
    \draw (0,5.2) node {Delay};
    \draw [smooth,samples=100,domain=2:8,color=blue,thick] plot(\x,1/2+8/\x);
    \draw [color=red,thick] (2,1/2+8/2) -- (4,1/2+8/4) -- (6,1/2+8/6) -- (8,1/2+8/8);


    \draw (2,-0.2) -- (2,0.2);
    \draw (4,-0.2) -- (4,0.2);
    \draw (6,-0.2) -- (6,0.2);
    \draw (8,-0.2) -- (8,0.2);

    \draw (2,-0.5) node {$1$};
    \draw (4,-0.5) node {$1+\frac{1}{k}$};
    \draw (6,-0.5) node {$1+\frac{2}{k}$};
    \draw (8,-0.5) node {$1+\frac{3}{k}$};

    \draw [color=red,thick] (6.5,5) -- (7,5);
    \draw (8.85,5) node {Asymptotic service times};
    \draw [color=blue,thick] (6.5,4.5) -- (7,4.5);
    \draw (8,4.5) node {Lower bound};

    \draw [dotted,thick] (2,0.2) -- (2,1/2+8/2);
    \draw [dotted,thick] (4,0.2) -- (4,1/2+8/4);
    \draw [dotted,thick] (6,0.2) -- (6,1/2+8/6);
    \draw [dotted,thick] (8,0.2) -- (8,1/2+8/8);
  \end{tikzpicture}
  \label{fig:upperVsLower}
  \caption{The asymptotic service time of the DQ policy vs. the corresponding lower bound, as functions of $r^*(\lambda)$.}
\end{figure}
\end{center}
On the other hand, if $k_n\to\infty$ as $n\to\infty$, it is easily shown that the lower bound, as $n\to\infty$, is
\[ \lim\limits_{n\to\infty} \left[ 1+ \frac{1}{\mu}\sum\limits_{i=1}^{k_n} \frac{1}{ r^* k_n -i+1} \right] = 1+ \frac{1}{\mu}\log\left( \frac{r^*}{r^*-1} \right), \]
which coincides with the asymptotic service time of the DQ policy given in Equation \eqref{eq:upperInf}.

\end{remark}

\begin{remark}
  In both theorems, the queueing delay $W_n^q$ converges to zero as $n\to\infty$. However, while in Theorem \ref{thm:UpperBound1} it converges to zero in expectation, in Theorem \ref{thm:UpperBoundK} it converges to zero in probability. This weaker result is an artifact of our analysis, and we conjecture that convergence in expectation holds in all cases.
\end{remark}

{
\begin{remark}
  We conjecture that even simpler policies such as Join-Idle-Queue \cite{joinIdleQueue}, are also asymptotically optimal. However, Join-Idle-Queue is surprisingly hard to analyze in this setting.
\end{remark}

\subsection{Main takeaways}
In this section, we obtained a universal lower bound on the expected delay of a typical job, and designed and analyzed dispatching/replication policies that match the universal lower bound, at least in the limit as $n\to\infty$, as long as $r^*k_n$ either is an integer, or diverges as $n\to\infty$. This establishes both the asymptotic optimality of our policies, as well as the tightness of our lower bound, in those cases. It is unclear whether our policies are always asymptotically optimal. However, there appears to be some slack in certain inequalities in the proof of the delay lower bound (Theorem \ref{thm:boundGeneral}), which suggests that our policies may indeed be always asymptotically optimal.

It is worth noting that the average number of replicas created per task under our policies is always equal to $r^*=1/\lambda-1/\mu$. In particular, this is independent of the task size distribution, of the realization of the task sizes, and of the current state of the queues. Combined with relatively simple dispatching rules for the replicas, this makes the aforementioned policies simple enough to be practical.

Finally, note that the lower bounds obtained in Theorem \ref{thm:boundGeneral} and Corollary \ref{cor:largeK} are increasing in the arrival rate $\lambda$. Since the lower bounds ignore possible queueing delays, the fact that they are increasing in $\lambda$ is solely due to the fact that higher arrival rates allow less overall replication (i.e., higher arrival rates lead to smaller $r^*$). Thus, this arrival rate versus delay performance tradeoff is a fundamental limitation of this kind of systems, and not a reflection of the usual effect of congestion on delays.

%
%
%
%
%
%
%
%

}

\section{General size-based slowdowns}\label{sec:sizeBased}
In this section, we introduce a new model for the slowdowns, and obtain asymptotically optimal policies within a somewhat restricted set of admissible policies.\\

For certain applications, the assumption that the slowdown  is independent from the task size is not realistic. Indeed, if replicas are very large, one can argue that the server side variability is ``averaged out'', which then results in less variable service times. By the same argument, if replicas are very small, it is more likely for them to complete their service within a time period when the server is atypically slow or fast, and service times become more variable. This intuition is somewhat validated by the fact that the higher variability of smaller tasks is widely observed in cloud computing systems operated by Facebook and Microsoft \cite{stoicaClones,stoicaIdleResources}.

In order to model these size-based slowdowns, we introduce an appropriate dependence structure between the slowdowns and the replicas' sizes.

\begin{assumption}\label{ass:sizeBasedSlowdowns}
For $r\geq 1$ and for $i\leq r$, let $S_{[i:r]}$ be the $i$-th order statistic of slowdowns $S_{1},\dots,S_{r}$ corresponding to a single job with replicas of size $X$. We assume the following.
\begin{enumerate}
\item For all $r \geq 1$, and all $x\geq 0$, we have
\[ \mathbb{E}\big[S_r \,\big|\, X=x \big]=\frac{1}{\mu}. \]
\item For all $r\geq 2$ and $i\leq r$, the expression
\[ \mathbb{E}\big[S_{[i:r]}\,\big|\, X=x\big] - \frac{1}{\mu} \]
is either increasing or decreasing in $x$. Furthermore,
\[ \lim\limits_{x\to\infty} \mathbb{E}\big[S_{[i:r]}\,\big|\, X=x\big] = \frac{1}{\mu}, \]
for all $r\geq 2$ and $i\leq r$.
\item For all $r\geq 1$ and $i\leq r$, the expression
\[ \mathbb{E}\big[S_{[i:r]}\,\big|\, X=x \big] - \mathbb{E}\big[ S_{[i:r+1]}\,\big|\, X=x\big] \]
is decreasing in $x$.
\end{enumerate}
\end{assumption}

Part 1 states that the expectation of the slowdowns does not depend on the task sizes. Part 2 reflects the premise that bigger tasks observe less variable slowdowns, and thus their order statistics converge monotonically to their mean. For instance, we might be dealing with a situation where the variance of the slowdowns converges to $0$ as $x\to\infty$. Finally, Part 3 implies that the improvement brought by an extra replica is smaller for larger tasks.

{
\begin{example}
  Suppose that the slowdowns have a Gamma distribution with shape parameter $x$, and rate parameter $\mu/x$. It is straightforward (but tedious) to check that these slowdowns satisfy Assumption \ref{ass:sizeBasedSlowdowns}.
\end{example}
}

\subsection{Block policies}
In order to keep the problem tractable, we restrict ourselves to a subset of the admissible policies introduced in Subsection \ref{sec:assumptions}, those that satisfy the following assumption.

\begin{assumption}[Block policy]\label{ass:block}
All the replicas associated with the same job start their service at the same time.
\end{assumption}
For the case $k_n=1$, this assumption can be satisfied by sending replicas to different queues in a careful way so that it is guaranteed that they will start their service at the same time. An example of this is the FREC policy introduced in Subsection \ref{sec:upperBoundFinite}, which is shown to be a Block policy in Lemma~\ref{lem:startFinish}. On the other hand, Assumption \ref{ass:block} excludes policies where the start times of the replicas are staggered in time, which could yield better performance in some cases.\\



Unfortunately, analyzing the class of general admissible policies appears to be difficult. For this reason, we restrict ourselves to Block policies, which are easier to analyze.

\subsection{Delay lower bound}
In this subsection, we present a straightforward lower bound on the expected delay of a typical job under Block policies, as the solution of a minimization problem. In order to do this, we consider a relaxation of the problem in which we have an infinite number of servers, but where there can only be $n$ busy servers in expectation, in steady-state.\\

First, note that the delay of a typical job is trivially lower bounded by the service time of a typical job. Since all replicas associated with the same job start their service at the same time (Assumption \ref{ass:block}), and since no replica with elapsed time can be prematurely cancelled (Assumption \ref{ass:noCancelInService}), the service time of a typical job under a Block policy is completely determined by the (possibly random) number of replicas that start their service. In particular, the service time of a typical job is determined by a (policy specific) measurable function $p:\mathbb{R}_+\to[0,1]^\infty$, where $p_r(x)$ is the probability that a typical job with tasks of size $x$ has $r$ replicas associated with it that get created and start their service. Using this, we obtain a lower bound on the expected delay of a typical job under a Block policy, as the solution of an optimization problem.

\begin{lemma}\label{lem:sizeBasedLowerBound}
Fix some $n$, and consider an admissible policy (i.e., a policy for which the process $\mathcal{Q}_n(\cdot)$ satisfies assumptions \ref{ass:markovianity} and \ref{ass:noCancelInService}), that satisfies Assumption~\ref{ass:block}, and under which the process $\mathcal{Q}_n(\cdot)$ admits at least one invariant probability measure $\pi_n$. Then, under Assumption \ref{ass:sizeBasedSlowdowns}, we have
\begin{align}
    \mathbb{E}_{\pi_n}[W_n]\geq \inf_{p\in \mathcal{P}_{k_n}} \quad &\int\limits_0^\infty x \sum\limits_{r=k_n}^\infty p_r(x)\Big(1+\mathbb{E}\big[ S_{[k_n:r]} \,\big|\, X=x \big] \Big) d\mathbb{P}_X(x) \label{eq:problem1} \\
    s.t.\quad & \frac{\lambda}{k_n} \int\limits_0^\infty x \sum\limits_{r=k_n}^\infty p_r(x) \Big[ r + (r-k_n)\mathbb{E}\big[S_{[k_n:r]}\,\big|\, X=x\big] \label{eq:problem2} \\
     &\qquad\qquad\quad \left. + \sum\limits_{i=1}^{k_n} \mathbb{E}\big[ S_{[i:r]}\,\big|\, X=x \big] \right] d\mathbb{P}_X(x) \leq 1, \nonumber
\end{align}
where $\mathbb{P}_X$ is the distribution of the task sizes, and $\mathcal{P}_{k_n}$ is the set of measurable functions from $\mathbb{R}_+$ to the infinite-dimensional simplex
\[ P_{k_n} \triangleq \left\{p\in[0,1]^\infty : \sum\limits_{r=k_n}^\infty p_r=1\right\}. \]
\end{lemma}
The proof is given in Appendix \ref{app:problemFormulation}.\\

{
The optimization problem defined by equations \eqref{eq:problem1} and \eqref{eq:problem2} corresponds to the minimization of the service time of a typical job over all possible distributions $p$ for the number of replicas that are created and start service, subject to the constraint that the average number of busy servers must be at most $n$.
}

In general, there is no closed form solution for the optimization problem in Lemma \ref{lem:sizeBasedLowerBound}. { However, in the following subsection we will derive some properties of optimal solutions in order to design delay optimal policies and understand their behavior.

\subsection{Properties of optimal solutions}

If $p^*$ is a feasible solution of the optimization problem defined by equations \eqref{eq:problem1} and \eqref{eq:problem2} that attains the minimum, we are interested in characterizing how it depends on $x$, and $k_n$.

To begin with, depending on the parameters of the system, the optimization problem defined by equations \eqref{eq:problem1} and \eqref{eq:problem2} may not have feasible solutions. A sufficient condition for the existence of feasible solutions is established in the following lemma.}

\begin{lemma}\label{lem:feasibility}
  If
  \begin{equation}\label{eq:stability2}
   \lambda\leq \frac{1}{1+\frac{1}{\mu}},
  \end{equation}
  then the optimization problem defined by equations \eqref{eq:problem1} and \eqref{eq:problem2} has a finite minimum, which is attained.
\end{lemma}
The proof is given in Appendix \ref{app:feasibility}.\\

Note that this sufficient condition for feasibility does not depend on $n$, $k_n$, nor on the task size distribution.

\begin{remark}
  Surprisingly, the condition of Equation \eqref{eq:stability2}  for the existence of a feasible solution is the same as the one in Theorem \ref{thm:stability} for the existence of a stable policy, which concerned the model with independent exponential slowdowns. However, unlike Theorem \ref{thm:stability}, this condition may not be necessary under this model.
\end{remark}

\begin{remark}
Note that any feasible point $\tilde{p}$ such that $\tilde{p}(x)=p^*(x)$ for almost every $x$ (with respect to $\mathbb{P}_X$) is also optimal. Therefore, any necessary properties of the optimal solutions are only true almost everywhere with respect to $\mathbb{P}_X$.
\end{remark}

%

{ The following result states some properties of the optimal solutions for arbitrary values of $k_n$.}

\begin{theorem}\label{thm:sizeBasedK}
Every feasible point $p$ that achieves the minimum in the problem defined by equations \eqref{eq:problem1} and \eqref{eq:problem2} is equal almost everywhere (with respect to $\mathbb{P}_X$) to a point $p^*$ that satisfies the following:
\begin{itemize}
  \item [(i)] $p^*_{k_n}(x) = 1$ for all $x$ large enough.
  \item [(ii)] There exists $\overline{r}_{k_n}\in O(k_n)$ such that $p_r^*(x)=0$ for all $r\geq \overline{r}_{k_n}$ and $x\geq 0$.
\end{itemize}
\end{theorem}
The proof is given in Appendix \ref{proof:sizeBasedK}.

\begin{remark}
  This result shows us that smaller tasks should be replicated more (although only up to a constant number of replicas per task), and that very large tasks should not be replicated at all. This conclusion is consistent with and provides support for current practice \cite{stoicaClones,stoicaIdleResources}. Note however that this conclusion can only be reached when the slowdowns depend on the task size; it does not hold under independent slowdowns, as in the $S\& X$ model \cite{decoupledSlowdown}.
\end{remark}

\subsubsection{Case $k_n=1$}
For the case where $k_n=1$ for all $n$, we can refine the results in Theorem \ref{thm:sizeBasedK} under the following additional assumption on the distribution of the slowdowns.

\begin{assumption}\label{ass:convexity}
 The expression
 \[ r\mathbb{E}\big[S_{[1:r]}\,\big|\, X=x\big] \]
 is convex as a function of the integer parameter $r$, for all $x\geq 0$.
\end{assumption}

The following lemma provides a sufficient condition for this assumption to hold.
\begin{lemma}
  If the expression
  \begin{equation}\label{eq:sufficientCondition}
   \mathbb{P}\left( \left. S > \frac{s}{r} \,\right|\, X=x \right)^{r}
  \end{equation}
  is convex as a function of the integer parameter $r$, for all $s,x>0$, then Assumption \ref{ass:convexity} holds.
\end{lemma}

For example, the exponential and the Pareto distributions satisfy Equation \eqref{eq:sufficientCondition}. In general, distributions with tails heavier than the exponential will also satisfy it, and therefore will also satisfy Assumption \ref{ass:convexity}.

\begin{theorem}\label{thm:sizeBased}
Suppose that $k_n=1$ for all $n$, and that the distributions of the slowdowns satisfy Assumption \ref{ass:convexity}. Then, every feasible point $p$ that achieves the minimum in the problem defined by equations \eqref{eq:problem1} and \eqref{eq:problem2} is equal almost everywhere (with respect to $\mathbb{P}_X$) to a point $p^*$ that satisfies the following:
\begin{itemize}
  \item [(i)] For every $x$, $p^*(x)$ is concentrated on up to two consecutive integers.
  \item [(ii)] If $\mathbb{P}_X$ is non-atomic, then $p^*(x)$ is concentrated in a single integer, for all $x$.
  \item [(iii)] The expected number of replicas,
  \[ r^*(x)\triangleq \sum\limits_{r=1}^\infty r p^*_r(x), \]
  is nonincreasing with $x$.
\end{itemize}
\end{theorem}
The proof is given in Appendix \ref{proof:sizeBasedThm}.

\begin{remark}
For the special case of $k_n=1$, Theorem \ref{thm:sizeBased} provides a crisper characterization of the optimal solutions $p^*$ than Theorem \ref{thm:sizeBasedK}. In particular, it states that the optimal number of replicas is a nonincreasing function of the task size.
\end{remark}

In the next subsection, we will use the properties derived in this subsection to design an asymptotically optimal Block policy.

\subsection{Asymptotically optimal Block policies}
In this subsection, we show that the lower bound on the delay obtained in the previous subsection is asymptotically attainable using an appropriate sequence of Block policies based on the FREC and DQ policies introduced in Subsection \ref{sec:upperBoundFinite}.\\

Suppose that
\[ \lambda < \frac{1}{1+\frac{1}{\mu}}. \]
Let $\alpha\in(1/2,1)$ be a constant, and let $p^{(n)}$ be a solution of the following optimization problem:
\begin{align*}
     \inf_{p\in \mathcal{P}_{k_n}} \quad &\int\limits_0^\infty x \sum\limits_{r=k_n}^\infty p_r(x)\Big(1+\mathbb{E}\big[ S_{[k_n:r]}\,\big|\, X=x \big] \Big) d\mathbb{P}_X(x)  \\
    s.t.\quad & \frac{\lambda}{k_n} \int\limits_0^\infty x \sum\limits_{r=k_n}^\infty p_r(x) \left[ r + (r-k_n)\mathbb{E}\big[S_{[k_n:r]}\,\big|\, X=x \big] {\color{white}\sum\limits_{i=1}^{k_n}} \right. \\
     &\qquad\qquad\qquad\qquad\qquad \left. + \sum\limits_{i=1}^{k_n} \mathbb{E}\big[ S_{[i:r]}\,\big|\, X=x \big] \right] d\mathbb{P}_X(x) \leq 1-n^{\alpha-1}.
\end{align*}
Note that, by moving the $1-n^{\alpha-1}$ to the other side, we get the same optimization problem as in Lemma \ref{lem:sizeBasedLowerBound}, with $\lambda$ replaced by $\lambda/(1-n^{\alpha-1})$. Furthermore, for all $n$ large enough, we have
\[ \frac{\lambda}{1-n^{\alpha-1}} \leq \frac{1}{1+\frac{1}{\mu}}. \]
Thus, Lemma \ref{lem:feasibility} applies, and a solution $p^{(n)}$ is guaranteed to exist. Moreover, recall that Theorem \ref{thm:sizeBasedK} states that there exists $\overline{r}_{k_n}^{(n)}\in O(k_n)$ such that $p_r^{(n)}(x)=0$ for all $r\geq \overline{r}_{k_n}^{(n)}$ and $x\geq 0$. In other words, $p^{(n)}$ is concentrated in its first $\overline{r}_{k_n}^{(n)}-k_n+1$ indices.

Let us partition the $n$-server system into $\overline{r}_{k_n}^{(n)}-k_n+1$ subsystems, indexed by $r=k_n,\dots,\overline{r}_{k_n}^{(n)}$, such that the $r$-th subsystem has
\begin{align*}
 n^{(r)} &\triangleq \left\lfloor \frac{\lambda n}{k_n} \int\limits_0^\infty x p^{(n)}_r(x) \left[ r + (r-k_n)\mathbb{E}\big[S_{[k_n:r]}\,\big|\, X=x \big] {\color{white}\sum\limits_{i=1}^{k_n}} \right. \right. \\
 &\left.\left. \qquad\qquad\qquad\qquad\qquad\qquad\quad + \sum\limits_{i=1}^{k_n} \mathbb{E}\big[ S_{[i:r]}\,\big|\, X=x \big] \right] d\mathbb{P}_X(x)  + \frac{n^{\alpha}}{\overline{r}_{k_n}^{(n)}} \right\rfloor
\end{align*}
servers. { As it was the case for the policies introduced in Subsection \ref{sec:upperBoundFinite}, the number of servers in the $r$-th subsystem is equal to the number of servers required to handle the incoming jobs to the subsystem, namely,
\begin{align*}
 &\frac{\lambda n}{k_n} \int\limits_0^\infty x p^{(n)}_r(x) \left[ r + (r-k_n)\mathbb{E}\big[S_{[k_n:r]}\,\big|\, X=x \big] {\color{white}\sum\limits_{i=1}^{k_n}} \right. \\
  &\qquad\qquad\qquad\qquad\qquad\qquad\qquad\qquad\qquad + \left. \sum\limits_{i=1}^{k_n} \mathbb{E}\big[ S_{[i:r]}\,\big|\, X=x \big] \right] d\mathbb{P}_X(x),
\end{align*}
plus a fraction of the ``spare'' servers,
\[ \frac{n^{\alpha}}{\overline{r}_{k_n}^{(n)}}. \]}
Without loss of generality, we assume that $n^{(r)}$ is a multiple of $r$, for all $r$. If this is not the case, we can always choose not to use some of the original $n$ servers.\\

We now introduce two policies. In both, when a job with tasks of size $x$ arrives to the system, it is sent to the $r$-th subsystem with probability $p_r^{(n)}(x)$. Then, we do one of the following.
\begin{itemize}
\item {\bf Size-Based Full Replication with Early Cancellation (SB-FREC)}: If at the time of an arrival to the $r$-th subsystem there are at least $r$ idle servers, then $r$ replicas are created and dispatched to idle servers. Otherwise, replicas are created and dispatched to \emph{all} servers in the subsystem. { When $r$ of these replicas have started their service, all other replicas associated with the same job are cancelled and immediately leave the system.}
\item {\bf Size Based Dummy Queues (SB-DQ)}: If at the time of an arrival to the $r$-th subsystem there are at least $r$ idle servers among the first $n^{(r)}- n^\alpha/2 \overline{r}_{k_n}^{(n)}$ servers in the subsystem, then $r$ replicas are dispatched to those idle servers. Otherwise, if there are at least $r$ idle servers among the last $n^\alpha/2\overline{r}_{k_n}^{(n)}$ servers of the $r$-th subsystem, then $r$ replicas are dispatched to those idle servers. Otherwise, jobs are queued in a virtual FIFO queue at the dispatcher until there are enough idle servers among the last $n^\alpha/2\overline{r}_{k_n}^{(n)}$ servers of the $r$-th subsystem.
\end{itemize}
Note that when there are enough idle servers, both of these policies send all replicas to idle servers. Thus, if the systems are in light enough heavy-traffic, the vast majority of replicas will be sent to idle servers in both cases. We also note that although the FREC policy seems superior, it appears to be tractable only for the case of $k_n=1$.\\

We now present the main results of this subsection.

\begin{theorem}\label{thm:sizeBasedThm2}
Suppose that $k_n=1$ for all $n$, and that
\[ \lambda < \frac{1}{1+\frac{1}{\mu}}. \]
Then, the SB-FREC policy is stable for all $n$ large enough, and we have
\[ \lim\limits_{n\to\infty} \mathbb{E}\big[W_n^s\big] = \int\limits_0^\infty x \sum\limits_{r=1}^\infty p^*_r(x)\Big(1+\mathbb{E}\big[S_{[1:r]}\,\big|\, X=x\big] \Big) d\mathbb{P}_X(x), \]
where $p^*$ is a solution of the optimization problem defined by equations \eqref{eq:problem1} and \eqref{eq:problem2}, for $k_n=1$. Furthermore, if
\[ \mathbb{E}\left[\Big(XS\Big)^{2+\epsilon}\right]<\infty \]
for some $\epsilon>0$, we also have
  \[ \lim\limits_{n\to\infty} \mathbb{E}\big[W_n^q\big] =0. \]
\end{theorem}
The proof is given in Appendix \ref{sec:proofSizeBasedUpper}.

\begin{theorem}\label{thm:sizeBasedThm4}
Suppose that $k_n=k$ for all $n$, and that
\[ \lambda < \frac{1}{1+\frac{1}{\mu}}. \]
Then, the SB-DQ policy is stable for all $n$ large enough, and we have
\[ \lim\limits_{n\to\infty} \mathbb{E}\big[W_n^s\big] = \int\limits_0^\infty x \sum\limits_{r=k}^\infty p^*_r(x)\Big(1+\mathbb{E}\big[S_{[k:r]}\,\big|\, X=x\big] \Big) d\mathbb{P}_X(x), \]
where $p^*$ is a solution of the optimization problem defined by equations \eqref{eq:problem1} and \eqref{eq:problem2}, for $k_n=k$. Furthermore, if
\[ \mathbb{E}\left[\Big(XS_{[k:k]}\Big)^{2+\epsilon}\right]<\infty \]
for some $\epsilon>0$, we also have
\[ \lim\limits_{n\to\infty} \mathbb{P}\big(W_n^q>0\big) =0. \]
\end{theorem}
The proof is given in Appendix \ref{sec:proofSizeBasedUpper}.

\begin{remark}
  Theorems \ref{thm:sizeBasedThm2} and \ref{thm:sizeBasedThm4} imply that the lower bound for the expected delay of a typical job established in Theorem \ref{lem:sizeBasedLowerBound} is asymptotically attained. Moreover, since the asymptotically delay optimal polices were constructed by using the solutions to a suitable version of the optimization problem that serves as the lower bound, the properties derived in theorems \ref{thm:sizeBasedK} and \ref{thm:sizeBased} (such as the fact that smaller tasks should be replicated more than large tasks) still hold and are therefore properties of our asymptotically delay-optimal policies.
\end{remark}

\subsection{Main takeaways}
In this section, we obtained a lower bound on the expected delay of a typical job under a restricted class of policies, as the solution of a certain optimization problem. Moreover, we analyzed the properties of its optimal solutions to guide the design of asymptotically optimal dispatching/replication policies.

Similar to the case of independent exponential slowdowns studied in Section \ref{sec:expLowerBound}, the number of replicas created per task is independent of the task size distribution, and of the current state of the queues. However, the number of replicas created for each incoming job now depends on the realization of the task sizes. This is a consequence of the dependence of the slowdowns on the size of the tasks (Assumption \ref{ass:sizeBasedSlowdowns}), which was not present in simpler models.

Finally, note that the constraint of the optimization problem that describes the asymptotic delay performance of our policies (Equation \eqref{eq:problem2}) depends on the arrival rate $\lambda$ in a way that a larger $\lambda$ can only increase the delay of the policy. Since the asymptotic delays of our policies are just their asymptotic service times, the fact that they are increasing in $\lambda$ is solely due to the fact that higher arrival rates allow less overall replication (i.e., higher arrival rates lead to smaller $r^*$). This was also observed in Section~\ref{sec:exponential} under the independent exponential slowdowns mode, which suggests that this arrival rate versus delay performance tradeoff is a fundamental limitation of replication-based systems, which transcends the model used for the distribution of the slowdowns.

\section{Conclusions and future work}\label{sec:conclusions}

The main objective of this paper was to study the impact of replication on the delay performance of distributed service systems with random server slowdowns. We did this under two different models for the server slowdowns, inspired by different applications.

When the server slowdowns are independent from the task sizes, we showed that the asymptotic expected delay is minimized when the number of replicas created per task is equal to a constant that only depends on the arrival rate, and on the expected slowdown. Surprisingly, it is independent from the number of tasks per job, and from the inter-arrival and task size distributions.

On the other hand, when the server slowdowns depend on the task sizes in some particular way, we showed that the asymptotic expected delay is minimized (among all policies where replicas associated with the same job start their service at the same time) when smaller tasks are replicated more than larger tasks.

There are several interesting directions for future research. For example:
\begin{itemize}
\item [(i)] In this paper, as well as in most of the literature, all servers have the same average processing rate. Thus, one interesting future line of work would be to explore the impact of heterogeneity in the average processing rates of servers in replication systems.
\item [(ii)] On a similar note, another interesting possible research direction would be to study the impact of heterogeneous jobs, where not all jobs can be served by the same number of servers (or at least not at the same speed).
    \item [(iii)] We conjecture that even simpler policies such as Join-Idle-Queue \cite{joinIdleQueue}, are also asymptotically optimal. However, Join-Idle-Queue is surprisingly hard to analyze in this setting.
\end{itemize}

%
%

\appendix

\section{Convenient notation}\label{sec:convenientNotation}
In this appendix, we introduce some definitions and notation that will be used in the proofs of theorems \ref{thm:stability} and \ref{thm:boundGeneral}.

For $j\geq 1$, we denote the {\bf relative service start times} of the replicas associated with the $j$-th job by $T_{j,r}$, ordered so that $0\leq T_{j,1}\leq\cdots\leq T_{j,n+k_n-1} \leq \infty$. More precisely, $T_{j,r}$ is the (random) time, relative to the arrival of the $j$-th job, when the $r$-th replica associated with it starts its service. As a convention, we set $T_{j,r}(\omega)=\infty$ if the $r$-th replica associated with the $j$-th job does not start its service under the sample path $\omega$. The times $T_{j,r}$ are mostly determined by the policy, but also depend on the queueing delays of replicas, and potentially on the service times of replicas associated with the same job that have already finished their service.

For $j\geq 1$ and for $i=1,\dots,k_n$, we denote the {\bf relative departure times} of the replicas associated with the $j$-th job by $D^{(i)}_j$, ordered so that $0\leq D^{(1)}_j\leq \cdots \leq D^{(k_n)}_j \leq \infty$. More precisely, $D^{(i)}_j$ is the time, relative to the arrival of the $j$-th job, when $i$ replicas associated with it finish their service. Since replicas cannot be cancelled after they start service (Assumption \ref{ass:noCancelInService}), $D^{(i)}_j$ is the $i$-th order statistic of the random variables
\[ \Big\{T_{j,r}+X_j\big(1+S_{j,r}\big)\Big\}_{r=1}^{n+k_n-1}. \]
In particular, $D^{(k_n)}_j$ is the delay of the $j$-th job. Moreover, we use the convention $D^{(0)}_j=0$.\\

The rest of the notation and definitions are only used for the proof of Theorem \ref{thm:boundGeneral}.

\subsection{Partition of the jobs in service}
We now introduce a dynamic partition of the set of jobs that are in service (i.e., the jobs that have at least one replica associated with it in service), by assigning a phase to each job in service. At any given time:
\begin{itemize}
  \item [a)] A job with tasks of size $x$ is in phase $0$ if at least one replica associated with it has started its service, and if all of the replicas associated with it have been in service for less than $x$ units of time.
  \item [b)] For $i=1,\dots,k_n$, a job with tasks of size $x$ is in phase $i$ if it is not in phase $0$ (i.e., if at least one replica associated with it, past or present, has been in service at least $x$ units of time) and if exactly $i-1$ replicas associated with it have finished their services.
\end{itemize}
These phases define a partition of the jobs that are in service, consisting of $k_n+1$ subsets (or phases). Moreover, since this is a dynamical system, jobs not only start and finish their services, but they also change phases. In particular, all jobs start their services in phase $0$, they can only leave the system if they are in phase $k_n$, and they go through all the intermediate phases in ascending order (possibly spending $0$ time in some phases). The possible transitions between the phases are depicted in Figure \ref{fig:jobTransitions}.

\begin{figure}[ht!]
  \begin{center}
    \begin{tikzpicture}[scale=0.5]

        \tikzset{node style/.style={state,
                                    minimum width=0.7cm,
                                    line width=0.3mm,
                                    fill=gray!20!white}}

        \node at (0,0.4) (start) {New};
        \node at (0,-0.4) (start) {job};
        \node[node style] at (4, 0)     (zero)     {$0$};
        \node[node style] at (8, 0)     (one)     {$1$};
        \node[node style] at (12, 0)     (two)     {$2$};
        \node[node style] at (20, 0)     (last)     {$k_n$};
        \node at (23,0) (exit) {Exit};

        \draw[
              auto=right,
              line width=0.3mm,
              >=latex]
            (1,0)      edge[->] node {} (zero)
            (zero)     edge[->] node {} (one)
            (one)     edge[->] node {} (two)
            (two)     edge[->] node {} (15,0)
            (17,0)     edge[->] node {} (last)
            (last)    edge[->] node {} (exit);
        \draw[dotted,very thick] (15.25,0) -- (16.75,0);
\end{tikzpicture}
  \end{center}
  \label{fig:jobTransitions}
  \caption{Possible transitions of jobs in service between the phases.}
\end{figure}
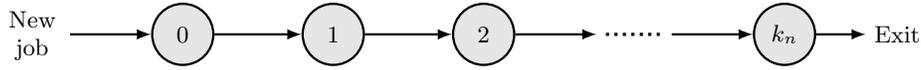

Suppose that the $j$-th job arrives to the system at time $0$. Then:
\begin{itemize}
  \item [a)] The $j$-th job is in phase $0$ between the time that the first replica associated with it starts its service ($T_{j,1}$), and the time that the first replica associated with it has spent $X_j$ units of time in service ($T_{j,1}+X_j$). At that point, the job moves to phase $1$.
  \item [b)] The $j$-th job is in phase $1$ between the time that the first replica associated with it is in service for $X_j$ units of time ($T_{j,1}+X_j$), and the first time that one of the replicas associated with it finishes its service. At that point, it moves to phase $2$.
  \item [c)] For $i=2,\dots,k_n$, the $j$-th job is in phase $i$ between the $(i-1)$-th time that one of the replicas associated with it finishes its service ($D^{(i-1)}_j$), and the $i$-th time that one of the replicas associated with it finishes its service ($D^{(i)}_j$). At that point, the job either moves to phase $i+1$ (if $i<k_n$) or leaves the system (if $i=k_n$).
\end{itemize}

\subsection{Partition of the replicas in service}
We now introduce similar terminology for the replicas that are in service. First, we define two types of replicas.
\begin{itemize}
  \item [a)] A replica of size $x$ is \emph{wasteful} if is in service less than $x$ units of time in total. In particular, the $r$-th replica of the $j$-th job is wasteful if and only if $D_j^{(k_n)}<T_{j,r}+X_j$.
  \item [b)] A replica of size $x$ is \emph{useful} if it is in service for at least $x$ units of time in total. In particular, the $r$-th replica of the $j$-th job is useful if and only if $D_j^{(k_n)}\geq T_{j,r}+X_j$.
\end{itemize}
Note that the type of a replica is future-dependent, noncausal property. As such, it does not change with the passage of time.

Second, we define a partition of the useful replicas in service by assigning a phase to each one of them, as follows. At any given time:
\begin{itemize}
  \item [a)] A useful replica of size $x$ is in phase $0$ if it has been in service less than $x$ units of time.
  \item [b)] For $i=1\dots,k_n$, a useful replica of size $x$ is in phase $i$ if it has been in service for at least $x$ units of time, and if it is associated with a job in phase $i$.
\end{itemize}
Note that this indeed defines a partition of all the useful replicas in service into $k_n+1$ subsets (or phases).  Moreover, as the system evolves, useful replicas change phases over time. In particular, the transitions between phases are as follows. Suppose that the $j$-th job arrives to the system at time $0$. The $r$-th replica associated with the $j$-th job starts its service at time $T_{j,r}$, as long as this time comes before the job leaves the system, i.e., only if $T_{j,r}<D^{(k_n)}_j$. Assuming that this is the case, and that the replica is useful (i.e., that $T_{j,r}+X_j\leq D^{(k_n)}_j$ and thus it spends at least $X_j$ units of time in service), we have the following.
\begin{itemize}
  \item [a)] The $r$-th replica associated with the $j$-th job is in phase $0$ between the time that it starts its service ($T_{j,r}$), and the time that it spends $X_j$ units of time in service ($T_{j,r}+X_j$). At that point, it moves to one of the phases $1,\dots,k_n$. In particular, it moves to phase $i$ if it is associated with a job in phase $i$.
  \item [b)] The $r$-th replica associated with the $j$-th job is in phase $1$ between the time when it has spent $X_j$ units of time in service ($T_{j,r}+X_j$), and the first time that a replica associated with the same job has finished its service ($D^{(1)}_j$). At that point, the replica either leaves the system (if it has finished its service), or moves to phase $2$ (if it has not finished its service).
  \item [c)] For $i=2,\dots,k_n$, the $r$-th replica associated with the $j$-th job is in phase $i$ between the time when it has spent $X_j$ units of time in service and $i-1$ replicas associated with the same job have already finished service ($\max\{T_{j,r}+X_j,D^{(i-1)}_j\}$), and the $i$-th time that a replica associated with the same job finishes its service ($D^{(i)}_j$). At that point, the replica either leaves the system (if it has finished its service or if $i=k_n$), or moves to phase $i+1$ (if it has not finished its service and $i<k_n$).
\end{itemize}
The possible transitions between phases are depicted in Figure \ref{fig:replicaTransitions}.

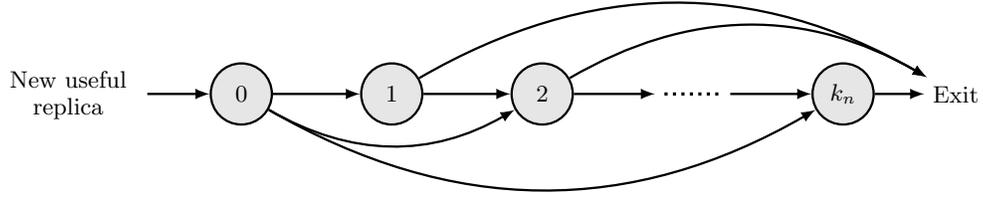
\begin{figure}[ht!]
  \begin{center}
    \begin{tikzpicture}[scale=0.5]

        \tikzset{node style/.style={state,
                                    minimum width=0.7cm,
                                    line width=0.3mm,
                                    fill=gray!20!white}}

        \node at (-0.6,0.4) (start) {New useful};
        \node at (-0.6,-0.4) (start) {replica};
        \node[node style] at (4, 0)   (zero)  {$0$};
        \node[node style] at (8, 0)   (one)   {$1$};
        \node[node style] at (12, 0)  (two)   {$2$};
        \node[node style] at (20, 0)  (last)  {$k_n$};
        \node at (23,0) (exit) {Exit};

        \draw[
              auto=right,
              line width=0.3mm,
              >=latex]
            (1.5,0)  edge[->]               node {} (zero)
            (one)   edge[->,bend left=30]  node {} (exit)
            (two)   edge[->,bend left=30]  node {} (exit)
            (zero)  edge[->,bend right=30] node {} (two)
            (zero)  edge[->,bend right=30] node {} (last)
            (zero)  edge[->]               node {} (one)
            (one)   edge[->]               node {} (two)
            (two)   edge[->]               node {} (15,0)
            (17,0)  edge[->]               node {} (last)
            (last)  edge[->]               node {} (exit);
        \draw[dotted,very thick] (15.25,0) -- (16.75,0);
\end{tikzpicture}
  \end{center}
  \label{fig:replicaTransitions}
  \caption{Possible transitions of replicas in service between the phases.}
\end{figure}

For $j\geq 1$, and for $i=1,\dots,k_n$, let $T_{j,r}^{(i)}$ be the {\bf relative phase start times} of the replicas in phase $i$ associated with the $j$-th job, arranged so that $0\leq T_{j,1}^{(i)} \leq \cdots \leq T_{j,n}^{(i)} \leq \infty$. More precisely, $T_{j,r}^{(i)}$ is the time elapsed between the moment when the $j$-th job starts phase $i$, and the moment when the $r$-th replica associated with it enters phase $i$. As a convention, we set $T^{(i)}_{j,r}(\omega)=\infty$ if the $r$-th replica associated with the $j$-th job is never in phase $i$ under the sample path $\omega$. These times are mostly determined by the policy, but also depend on the queueing delays of replicas, and potentially on the service times of replicas associated with the same job that have already finished their services.

Furthermore, for $j\geq 1$ and for $i=1,\dots,k_n$, let $S_{j,1}^{(i)},\dots, S_{j,n}^{(i)}$ be such that $X_j S_{j,1}^{(i)},\dots, X_j S_{j,n}^{(i)}$ are the {\bf remaining service times} of the replicas in phase $i$ associated with the $j$-th job. That is, $X_j S_{j,r}^{(i)}$ is the remaining service time of the $r$-th replica associated with the $j$-th job, when the replica starts phase $i$. Thus, for $i=1,\dots,k_n$, the time that the $j$-th job spends in the $i$-th phase is equal to
\[ \min\limits_{r=1,\dots,n} \left\{ T_{j,r}^{(i)}+X_j S_{j,r}^{(i)} \right\}. \]

%
%
%

We illustrate the concepts and definitions introduced above with the following example.

\begin{example}
Suppose that $k_n=2$. In Figure \ref{fig:phasesExample} we depict the evolution of the phases of the $j$-th job, and of its replicas, for a sample path of a certain policy. In our example, only the first four replicas that start their service are useful, while the fifth one is wasteful. Moreover, the fourth replica is never in phase $1$.

\begin{figure}[ht!]
\begin{center}
\begin{tikzpicture}[scale=0.8]
  \filldraw[black!10] (1,0) rectangle (13,1);
  \filldraw[black!10] (2,1) rectangle (9,2);
  \filldraw[black!10] (3,2) rectangle (13,3);
  \filldraw[black!10] (7,3) rectangle (13,4);
  \filldraw[black!10] (11,4) rectangle (13,5);

  \draw (1,0) rectangle (13,1);
  \draw (2,1) rectangle (9,2);
  \draw (3,2) rectangle (13,3);
  \draw (7,3) rectangle (13,4);
  \draw (11,4) rectangle (13,5);

  \draw[thick,->] (0,-0.5) -- (0,7);
  \draw (0.5,7.4) node {Replicas};
  \draw[thick,->] (-0.5,0) -- (14,0);
  \draw (14.8,0.25) node {Relative};
  \draw (14.8,-0.25) node {time};

  \draw (4.5,0) rectangle (9,1);
  \draw (5.5,1) rectangle (9,2);
  \draw (6.5,2) rectangle (9,3);
  \draw (10.5,3) rectangle (13,4);

  \draw (2.75,0.5) node {Phase $0$};
  \draw (3.75,1.5) node {Phase $0$};
  \draw (4.75,2.5) node {Phase $0$};
  \draw (8.75,3.5) node {Phase $0$};

  \draw (6.75,0.5) node {Phase $1$};
  \draw (7.25,1.5) node {Phase $1$};
  \draw (7.75,2.5) node {Phase $1$};

  \draw (11,0.5) node {Phase $2$};
  \draw (11,2.5) node {Phase $2$};
  \draw (11.75,3.5) node {Phase $2$};

  \draw (12,4.5) node {Wasteful};

  \draw[dotted,thick] (1,1) -- (1,6.5);
  \draw[dotted,thick] (2,0) -- (2,1);
  \draw[dotted,thick] (3,0) -- (3,2);
  \draw[dotted,thick] (4.5,1) -- (4.5,6.5);
  \draw[dotted,thick] (7,0) -- (7,3);
  \draw[dotted,thick] (9,3) -- (9,6.5);
  \draw[dotted,thick] (11,0) -- (11,4);
  \draw[dotted,thick] (13,1) -- (13,2);
  \draw[dotted,thick] (13,5) -- (13,6.5);

  \draw (1,-0.5) node {$T_{j,1}$};
  \draw (2,-0.5) node {$T_{j,2}$};
  \draw (3,-0.5) node {$T_{j,3}$};
  \draw (7,-0.5) node {$T_{j,4}$};
  \draw (11,-0.5) node {$T_{j,5}$};

  \draw (4.5,-0.5) node {$T_{j,1}+X_j$};

  \draw (9,-0.5) node {$D^{(1)}_j$};
  \draw (13,-0.5) node {$D^{(2)}_j$};

  \draw (1,-0.15) -- (1,0.15);
  \draw (2,-0.15) -- (2,0.15);
  \draw (3,-0.15) -- (3,0.15);
  \draw (4.5,-0.15) -- (4.5,0.15);
  \draw (7,-0.15) -- (7,0.15);
  \draw (11,-0.15) -- (11,0.15);
  \draw (9,-0.15) -- (9,0.15);
  \draw (13,-0.15) -- (13,0.15);

  \draw (2.75,6.25) node {Job is in};
  \draw (2.75,5.75) node {phase $0$};
  \draw (6.75,6.25) node {Job is in};
  \draw (6.75,5.75) node {phase $1$};
  \draw (11,6.25) node {Job is in};
  \draw (11,5.75) node {phase $2$};
\end{tikzpicture}
\caption{Evolution of phases for a particular sample path of a given policy.}
    \label{fig:phasesExample}
\end{center}
\end{figure}
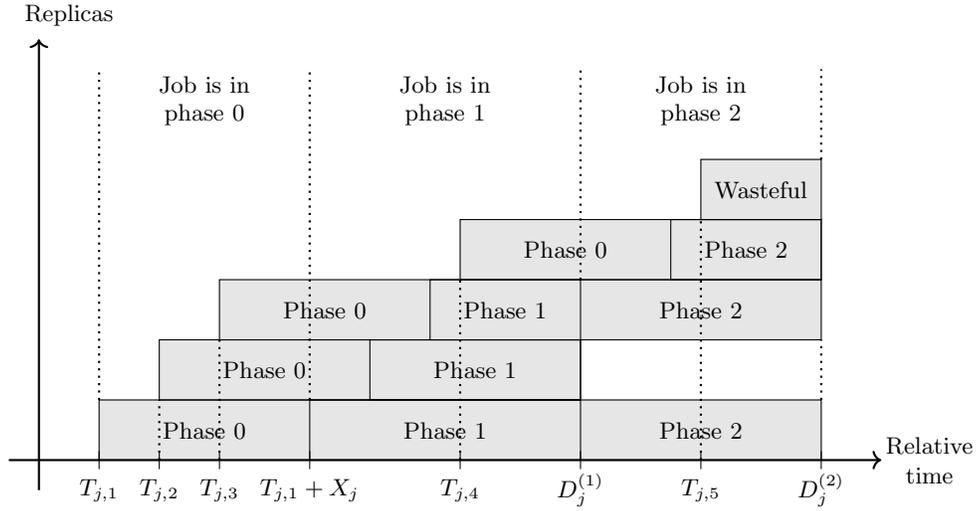
In the example, the relative phase start times of replicas, for phase $1$, are
\begin{align*}
 T_{j,1}^{(1)}&=0 \\
 T_{j,2}^{(1)}&=T_{j,2}-T_{j,1} \\
 T_{j,3}^{(1)}&=T_{j,3}-T_{j,1},
\end{align*}
and the relative phase start times of replicas, for phase $2$, are
\begin{align*}
 T_{j,1}^{(2)}&=0 \\
 T_{j,2}^{(2)}&=0 \\
 T_{j,3}^{(2)}&=T_{j,4}-D_j^{(1)}.
\end{align*}
\end{example}

\section{Proof of the necessary condition in Theorem 3.1}\label{proof:stability}
A necessary condition for the existence of an invariant probability measure for the extended queue state process $\mathcal{Q}_n(\cdot)$ is that the total workload in the queues does not diverge with time. In the parallel queueing model that we consider, the total workload can only remain bounded if the average rate at which workload comes into the queues is strictly less than the system's total processing power (see \cite{asmussen}, for example). This means that the arrival rate of jobs, times the expected workload of a typical job, must be strictly less than the total processing power of the system. Since the arrival rate is $\lambda n/k_n$, and the total processing power is $n$, the stability condition is
\begin{equation}
  \left( \frac{\lambda n}{k_n} \right)\big(\text{Expected workload of typical job}\big)<n. \label{eq:stabilityCondition}
\end{equation}
Note that the workload of a job in a replication system depends on the policy, as is defined as the sum of all the service times of the replicas associated with it. With this in mind, we now proceed to obtain a lower bound for the expected workload of a typical job under an admissible policy, using the notation introduced in Section \ref{sec:convenientNotation}.

Let $0\leq T_1\leq\cdots\leq T_{n+k_n-1}\leq\infty$ be the steady-state relative service start times of the replicas associated with a typical job with tasks of size $X$, and slowdowns $S_1,\dots,S_{n+k_n-1}$. For $i=k_n,\dots,n+k_n-1$ and for $r=1,\dots,k_n$, let $D_{[r:i]}$ be the $r$-th order statistic of the set of random variables
\[ \big\{T_1+X(1+S_1),\dots,T_i+X(1+S_i)\big\}. \]
Note that $D_{[1:i]}\leq\cdots\leq D_{[k_n:i]}$ are the relative departure times of a policy with relative service start times $0\leq\hat{T}_1\leq \cdots \leq \hat{T}_{n+k_n-1}\leq\infty$, such that $\hat{T}_r=T_r$, for all $r\leq i$, and $\hat{T}_r=\infty$, for all $r>i$.

For $i=k_n,\dots,n+k_n-1$, let
\[ M_i \triangleq \sum\limits_{r=1}^i \min\Big\{X(1+S_r),\,\,\big(D_{[k_n:i]}-T_r\big)^+\Big\}. \]
This is the workload of a typical job (i.e., the sum of the service times of all the replicas associated with it) under a policy with the relative service start times $0\leq\hat{T}_1\leq \cdots \leq \hat{T}_{n+k_n-1}\leq\infty$, such that $\hat{T}_r=T_r$, for all $r\leq i$, and $\hat{T}_r=\infty$, for all $r>i$. In particular, $M_{n+k_n-1}$ is the workload of a typical job under the original policy, so we are interested in finding a lower bound for it. With this in mind, consider the decomposition
\begin{equation}\label{eq:workloadDecomposition}
 M_{n+k_n-1}=M_{k_n}+\sum\limits_{i=k_n+1}^{n+k_n-1} \Big[ M_i-M_{i-1} \Big].
\end{equation}
The rest of the proof consists of obtaining an expression for the expectation of the first term on the right-hand side, and showing that the expectation of the rest of the terms in the right-hand side is nonnegative.\\

First, note that $T_r+X(1+S_r)\leq D_{[k_n:k_n]}$, for all $r\leq k_n$, and thus $X(1+S_r)\leq \big(D_{[k_n:k_n]}-T_r\big)^+$, for all $r\leq k_n$. It follows that
\[ M_{k_n} = \sum\limits_{r=1}^{k_n} X(1+S_r). \]
Combining this with the fact that $S_1,\dots,S_{k_n}$ are independent exponential random variables with mean $1/\mu$, independent from $X$ (Assumption \ref{ass:independentExponential}), we get that
\begin{equation}\label{eq:initialWorkload}
 \mathbb{E}[M_{k_n}] = \mathbb{E} \left[ \sum\limits_{r=1}^{k_n} X(1+S_r) \right] = k_n \left( 1+\frac{1}{\mu} \right).
\end{equation}

On the other hand, note that the difference of workloads $M_i-M_{i-1}$ is how much more the servers have to work when up to $i$ replicas are ever made, compared to when only up to $i-1$ replicas are ever made. It is not immediately clear whether this quantity is positive or not. If the $i$-th replica indeed starts its service at some point, it adds workload to its assigned server. However, if the $i$-th replica finishes its service before $k_n$ of the first $i-1$ finish theirs, the $i$-th replica contributes to the premature cancellation of other replicas, effectively reducing the workloads of their servers. We will now show that, in expectation, the difference in workload is nonnegative.\\

Let us fix $i$, and consider the following partition of the sample space. Let $E_1$ be the event that $k_n$ replicas finish their service before the $i$-th replica spends $X$ units of time in service, i.e., the event that $D_{[k_n:i-1]} \leq T_i+X$. Moreover, for $\ell=0,\dots,k_n-1$, let $E_2(\ell)$ be the event that $\ell<k_n$ replicas finished their service by the time the $i$-th replica spends $X$ units of time in service, i.e., the event that $D_{[\ell:i-1]}\leq T_i+X < D_{[\ell+1:i-1]}$ (using the convention that $D_{[0:i-1]}=0$). We show that the expectation of the difference in workloads is nonnegative, for each of the events defined above, in the following two lemmas.

\begin{lemma}\label{lem:workload1}
If $\mathbb{P}(E_1)>0$, then $\mathbb{E}\big[M_i-M_{i-1} \mid E_1 \big]\geq 0$.
\end{lemma}
\begin{proof}
First note that in the event $E_1$, we have $D_{[k_n:i-1]} \leq T_i+X(1+S_i)$ and thus $D_{[k_n:i-1]}=D_{[k_n:i]}$. It follows that
    \begin{align*}
     \mathbb{E}\big[M_i-M_{i-1}\mid E_1\big] &= \mathbb{E} \left[ \left. \sum\limits_{r=1}^i \min\Big\{X(1+S_r),\,\,\big(D_{[k_n:i]}-T_r\big)^+\Big\} \,\right|\, E_1 \right] \\
     &\qquad - \mathbb{E} \left[ \left. \sum\limits_{r=1}^{i-1} \min\Big\{X(1+S_r),\,\,\big(D_{[k_n:i-1]}-T_r\big)^+\Big\} \,\right|\, E_1 \right]\\
     &=\mathbb{E} \left[ \left. \min\Big\{X(1+S_i),\,\,\big(D_{[k_n:i]}-T_i\big)^+\Big\} \,\right|\, E_1 \right] \\
     &\geq 0.
    \end{align*}
\end{proof}

\begin{lemma}\label{lem:workload2}
  If $\mathbb{P}\big(E_2(\ell)\big)>0$, then $\mathbb{E}\big[M_i-M_{i-1}\mid E_2(\ell)\big]\geq 0$.
\end{lemma}
\begin{proof}
In order to show that $\mathbb{E}[M_i-M_{i-1}\mid E_2(\ell)]\geq 0$, we will rewrite $M_i$ and $M_{i-1}$ as the sum of the workload processed before and after time $T_i+X$ in the event $E_2(\ell)$, and show that the conditional expectation of the difference in workloads is nonnegative both before and after time $T_i+X$.

On the one hand, the expected workload processed before time $T_i+X$ conditioned on the event $E_2(\ell)$ can be expressed as the sum of the expected elapsed service time of each replica at time $T_i+X$ conditioned on the event $E_2(\ell)$, that is
\[ \mathbb{E}\left[\left. \sum\limits_{r=1}^{i-1} \min\big\{X(1+S_r),(T_i+X)-T_r\big\} \,\right|\, E_2(\ell) \right] \]
when $i-1$ replicas are created, and
\[ \mathbb{E}\left[\left. X+ \sum\limits_{r=1}^{i-1} \min\big\{X(1+S_r),(T_i+X)-T_r\big\} \,\right|\, E_2(\ell) \right] \]
when $i$ replicas are created. Consequently, the expected difference in workload processed before time $T_i+X$ is
\begin{equation*}\label{eq:exp_workload_initial}
  \mathbb{E}[X \mid E_2(\ell)] \geq 0.
\end{equation*}

On the other hand, the expected workload processed after time $T_i+X$ conditioned on the event $E_2(\ell)$ can be expressed as the sum of the expected times that replicas spend between consecutive departures of replicas after time $T_i+X$ conditioned on the event $E_2(\ell)$, multiplied by the number of replicas in service during that time, that is
\begin{align}
 &\mathbb{E}\left[ (i-1-\ell)\left[ D_{[\ell+1:i-1]} - (T_i+X) \right] {\color{white}\sum\limits_{r=\ell+1}^{k_n-1}} \right. \nonumber \\
     &\qquad\qquad\qquad\qquad \left.\left. + \sum\limits_{r=\ell+1}^{k_n-1} (i-1-r)\left( D_{[r+1:i-1]} -D_{[r:i-1]} \right) \right| E_2(\ell) \right] \label{eq:workload1}
\end{align}
when $i-1$ replicas are created, and
\begin{align}
 &\mathbb{E}\left[ (i-\ell)\left[ D_{[\ell+1:i]} - (T_i+X) \right] {\color{white}\sum\limits_{r=\ell+1}^{k_n-1}} \right. \nonumber \\
     &\qquad\qquad\qquad\qquad\qquad\qquad \left.\left. + \sum\limits_{r=\ell+1}^{k_n-1} (i-r)\left( D_{[r+1:i]} -D_{[r:i]} \right) \right| E_2(\ell) \right] \label{eq:workload2}
\end{align}
when $i$ replicas are created. We will now argue that these two expressions are equal.

Consider the case where $i$ replicas are made. Recall that the extended queue state process $\mathcal{Q}_n(\cdot)$, which keeps track of the size of the replicas and of their elapsed service times but not of their remaining service times, is Markov (Assumption \ref{ass:markovianity}). Moreover, recall that slowdowns are independent and exponentially distributed with mean $1/\mu$, and independent from the task size $X$ (Assumption \ref{ass:independentExponential}). Combining these two facts it follows that, conditioned on $E_2(\ell)$, the remaining service times of the $i-\ell$ replicas that are still in service at time $T_i+X$ are $X\tilde S_1,\dots,X\tilde S_{i-\ell}$, where $\tilde S_1,\dots,\tilde S_{i-\ell}$ are independent and exponentially distributed random variables with mean $1/\mu$, independent from $X$. Moreover, note that for $r=\ell+1,\dots,k_n$, $D_{[r,i]}-(T_i+X)$ is the $r$-th order statistic of the remaining service times $X\tilde S_1,\dots,X\tilde S_{i-\ell}$. It follows that, conditioned on $X=x$ and $E_2(\ell)$, the random variables $D_{[\ell+1:i]} - (T_i+X)$, and $D_{[r+1:i]} -D_{[r:i]}$ are exponential with rates $\mu(i-\ell)/x$ and $\mu(i-r)/x$, respectively. Combining this with Equation \eqref{eq:workload2} we obtain
\begin{align*}
     &\mathbb{E}\left[ \left. (i-\ell)\left[ D_{[\ell+1:i]} - (T_i+X) \right] + \sum\limits_{r=\ell+1}^{k_n-1} (i-r)\left( D_{[r+1:i]} -D_{[r:i]} \right) \right| E_2(\ell) \right] \\
     &\qquad=\mathbb{E}\left[ \mathbb{E}\left[ (i-\ell)\left[ D_{[\ell+1:i]} - (T_i+X) \right] {\color{white}\sum\limits_{r=\ell+1}^{k_n-1}} \right.\right. \\
     &\qquad\qquad\qquad\qquad\quad\,\, \left.\left.\left.\left.+ \sum\limits_{r=\ell+1}^{k_n-1} (i-r)\left( D_{[r+1:i]} -D_{[r:i]} \right) \right| X,\, E_2(\ell) \right] \right| E_2(\ell) \right] \\
     &\qquad=\mathbb{E}\left[ \left. (i-\ell)\left[ \frac{X}{\mu(i-\ell)}  \right] + \sum\limits_{r=\ell+1}^{k_n-1} (i-r)\left[ \frac{X}{\mu(i-r)} \right] \right| E_2(\ell) \right] \\
     &\qquad=\left(\frac{k_n-\ell}{\mu}\right) \mathbb{E}\big[X \,\big|\, E_2(\ell)\big].
\end{align*}
Similarly, it can be checked that the expected workload processed after time $T_i+X$ when only $i-1$ replicas are created (Equation \eqref{eq:workload1}) is the same.
\end{proof}

Since the events $E_1$ and $E_2(0),\dots,E_2(k_n-1)$ form a partition, lemmas \ref{lem:workload1} and \ref{lem:workload2} imply that
\begin{align*}
 \mathbb{E}[M_i-M_{i-1}] &= \mathbb{E}[M_i-M_{i-1}\mid E_1] \mathbb{P}(E_1) \\
 & \qquad\qquad\qquad\qquad\quad  + \sum\limits_{r=0}^{k_n-1} \mathbb{E}[M_i-M_{i-1}\mid E_2(r)] \mathbb{P}(E_2(r))
\end{align*}
is nonnegative, for all $i$. Combining this with equations \eqref{eq:workloadDecomposition} and \eqref{eq:initialWorkload}, the expected workload of a typical job is
\[ \mathbb{E}[M_{n+k_n-1}]\geq k_n \left( 1+\frac{1}{\mu} \right). \]
Finally, combining this with the stability condition in Equation \eqref{eq:stabilityCondition}), we obtain the necessary condition for stability
\[ \left(\frac{\lambda n}{k_n}\right) k_n\left( 1+\frac{1}{\mu} \right) < n, \]
which is equivalent to
\[ \lambda < \frac{1}{1+\frac{1}{\mu}}. \]

\section{Proof of Theorem 3.2}\label{proof:lowerBound}
We consider a relaxation of the problem in which there are an infinite number of servers available, but there can only be up to $n$ busy servers in expectation, in steady-state, and there can only be up to $n$ replicas associated with the same job in service, at any point in time. Since this includes policies that have only up to $n$ busy servers at all times, any lower bound for the delay in this relaxed setting is also a lower bound for the delay in the original setting. Thus, we shall prove a lower bound for this infinite-server relaxation. First, we will obtain a lower bound for the case where there are only finitely many task sizes, and then we generalize it to task sizes with general distributions via a comparison argument.

\subsection{Finitely many task sizes}

Consider the case where the task size $X$ can only take values in a finite set $\mathcal{X}$, with $\mathbb{E}[X]=\eta$ (which is a further relaxation of the original assumption that $\mathbb{E}[X]=1$). For each $x\in\mathcal{X}$, let $p(x)\triangleq\mathbb{P}(X=x)$ be the probability that a job has tasks of size $x$.

Let us fix a stable admissible policy, and consider the system in steady-state. For each $x\in\mathcal{X}$, let $\tilde{p}(x)$ be the expected number of servers (normalized by $n$) that are working on replicas of size $x$, in steady-state. In this setting, the proof is completed in three steps:
\begin{enumerate}
  \item For each $x\in\mathcal{X}$, we obtain a lower bound for the expected time that a typical \emph{replica} of size $x$ spends in each phase, as a function of $\tilde p(x)$.
  \item For each $x\in\mathcal{X}$, we show that the expected time that a typical \emph{replica} of size $x$ spends in each phase is smaller than or equal to expected time that a typical \emph{job} with replicas of size $x$ spends in the same phase. Combining this with the first step, and adding up the lower bounds for all phases and task sizes, we obtain a lower bound for the expected delay of a typical job as a function of $\{\tilde p(x):x\in\mathcal{X}\}$.
  \item We minimize the lower bound obtained in the previous step with respect to $\{\tilde p(x):x\in\mathcal{X}\}$, which yields a lower bound on the expected delay of a typical job that only depends on the parameters of the system.
\end{enumerate}

\begin{lemma}\label{lem:first}
  For every $x\in\mathcal{X}\backslash\{0\}$, and for $i=1,\dots,k_n$, we have
\[ \mathbb{E}\left[ \left. \overline{W}_n^{(i)} \right| X=x  \right] \geq \frac{x}{k_n\mu \left( \frac{\tilde{p}(x)}{\lambda p(x) x}- \frac{1}{\mu} \right) - \mu(i-1)}, \]
where $\overline{W}_n^{(i)}$ is the time that a typical replica of size $X$ spends in phase $i$.
\end{lemma}
\begin{proof}
We first fix some $x\in\mathcal{X}\backslash\{0\}$, and introduce some notation. For $i=0,\dots,k_n$,
\begin{itemize}
\item [(i)] let $N_n^{(i)}(x)$ be the steady-state number of replicas of size $x$, that are in phase $i$,
\item [(ii)] let $\lambda_n^{(i)}(x)$ be the arrival rate of replicas of size $x$, to phase $i$.
\end{itemize}

Recall that the extended queue state process $\mathcal{Q}_n(\cdot)$, which keeps track of the size of the replicas and of their elapsed service times but not of their remaining service times, is Markov (Assumption \ref{ass:markovianity}). Combining this with the fact that slowdowns are independent and exponential with mean $1/\mu$ (Assumption \ref{ass:independentExponential}), and with the fact that a job either changes phases or leaves the system whenever a replica finishes its service, we see that the departure rate of jobs with tasks of size $x$ from phase $i>0$ is equal to
\[ \frac{\mu}{x}\mathbb{E}\left[N_n^{(i)}(x)\right]. \]
Furthermore, since all jobs go through all phases, then the arrival rate of jobs with tasks of size $x$ to each phase is the same as the arrival rate of jobs with tasks of size $x$ to the system, which is $\lambda n p(x)/k_n$. Since the system is in steady-state, the arrival and exit rates of jobs from each phase must be the same, and thus
\[ \frac{\mu}{x}\mathbb{E}\left[N_n^{(i)}(x)\right]=\frac{\lambda p(x) n}{k_n}, \]
for $i=1,\dots,k_n$. Equivalently, we have
\begin{equation}\label{eq:N2}
  \mathbb{E}\left[N_n^{(i)}(x)\right]=\frac{\lambda p(x) xn}{k_n\mu},
\end{equation}
for all $i=1,\dots,k_n$. In particular, this means that the expected number of replicas in phases $1$ through $k_n$ is independent from the policy and from the phase.

On the other hand, since $\tilde{p}(x)$ is the expected fraction of servers that are working on replicas of size $x$, we have
\[ \sum\limits_{i=0}^{k_n} \mathbb{E}\left[N_n^{(i)}(x)\right] \leq \tilde{p}(x)n. \]
The inequality above is an equality only if there are no wasteful replicas of size $x$, i.e., replicas that spend less than $x$ units of time in service. Combining this with Equation \eqref{eq:N2}, we obtain
\begin{equation}\label{eq:N0}
 \mathbb{E}\left[N_n^{(0)}(x)\right] \leq n\left( \tilde{p}(x)-\frac{\lambda p(x) x}{\mu} \right).
\end{equation}
Moreover, recall that all useful replicas of size $x$ spend exactly $x$ units of time in phase $0$, by definition. Thus,
\[ \mathbb{E}\left[\left. \overline{W}_n^{(0)}\right| X=x\right]=x. \]
Combining this with Equation \eqref{eq:N0}, and applying Little's law, we obtain
\begin{equation}\label{eq:lambda0}
 \lambda_n^{(0)}(x)= \frac{\mathbb{E}\left[N_n^{(0)}(x)\right]}{\mathbb{E}\left[ \left.\overline{W}_n^{(0)} \right| X=x \right]} \leq n\left( \frac{\tilde{p}(x)}{x}- \frac{\lambda p(x)}{\mu} \right).
\end{equation}

Recall that, for $i=1,\dots,k_n$, all the replicas that enter phase $i$ had to enter phase $0$ before. Furthermore, there are at least $i-1$ replicas associated with the same job that enter phase $0$ but do not enter phase $i$ (because they finished their service earlier). Combining these two facts, and using that the arrival rate of jobs with tasks of size $x$ is $\lambda p(n) n / k_n$, we obtain
\begin{equation}\label{eq:lambda_i}
 \lambda_n^{(i)}(x) \leq \lambda_n^{(0)}(x) - \frac{\lambda p(x) n}{k_n}(i-1),
\end{equation}
for $i=1,\dots,k_n$. Applying Little's law once more, and using equations \eqref{eq:lambda0} and \eqref{eq:lambda_i}, we get
\begin{align}
 \mathbb{E}\left[ \left. \overline{W}_n^{(i)} \right| X=x \right] &= \frac{\mathbb{E}\left[N_n^{(i)}(x)\right]}{\lambda_n^{(i)}(x)} \nonumber \\
 &\geq \frac{\mathbb{E}\left[N_n^{(i)}(x)\right]}{\lambda_n^{(0)}(x) - \frac{\lambda p(x) n}{k_n}(i-1)} \nonumber \\
  &\geq \frac{x}{k_n\mu \left( \frac{\tilde{p}(x)}{\lambda p(x) x}- \frac{1}{\mu} \right) - \mu(i-1)}, \label{eq:boundT}
\end{align}
for $i=1,\dots,k_n$.
\end{proof}

We have thus obtained a lower bound on the expected time that a typical replica of size $x$ spends in phase $i$. However, we need to obtain a lower bound on the expected time that a typical \emph{job} with tasks of size $x$ spends in phase $i$. We prove that the former is smaller than or equal to the latter in the following claim.

\begin{lemma}\label{lem:second}
  For every $x\in\mathcal{X}\backslash\{0\}$, and for $i=1,\dots,k_n$, we have
  \[ \mathbb{E}\left[\left. W_n^{(i)} \right| X=x\right] \geq \mathbb{E}\left[\left. \overline{W}_n^{(i)} \right| X=x \right], \]
  where $W_n^{(i)}$ is the time that a typical job with tasks of size $X$ spends in phase $i$.
\end{lemma}
\begin{proof}
Recall that the time that a typical job spends in phase $i$ is
\begin{align*}
 W_n^{(i)}&=\min\left\{T_1^{(i)}+XS_1^{(i)},\dots,T_n^{(i)}+XS_n^{(i)} \right\},
\end{align*}
where $T_1^{(i)},\dots,T_n^{(i)}$ are the relative phase start times of replicas in the $i$-th phase, and $XS_1^{(i)},\dots,XS_n^{(i)}$ are the remaining service times of the replicas in the $i$-th phase. {\color{black}Moreover, recall that the extended queue state process $\mathcal{Q}_n(\cdot)$, which keeps track of the size of the replicas and of their elapsed service times but not of their remaining service times, is Markov (Assumption \ref{ass:markovianity}), and that slowdowns are independent and exponentially distributed with mean $1/\mu$, and independent from $X$ (Assumption \ref{ass:independentExponential}). Combining these two facts we get that $S_1^{(i)},\dots,S_n^{(i)}$ are independent and exponentially distributed with mean $1/\mu$, and independent from $X$. It follows that, conditioned on $X=x> 0$, the hazard rate of $W_n^{(i)}$ at $t$ is equal to $\mu/x$ times the expected number of replicas in phase $i$ associated with a typical job, $t$ units of time after it has started its $i$-th phase, given that $W_n^{(i)}\geq t$. Namely, conditioned on $X=x> 0$, the hazard rate of $W_n^{(i)}$ at $t$ is
\begin{equation*}
  h_{W_n^{(i)}}(t \mid x) = \frac{\mu}{x}\sum\limits_{\ell=1}^n \mathbb{P}\left( \left. T^{(i)}_\ell < t \,\right|\,W_n^{(i)}> t,\, X=x \right).
\end{equation*}

Let $\overline{W}_n^{(i,r)}$ be the time that the $r$-th replica to enter phase $i$ is in service, conditioned on having entered phase $i$ at some point in time. In particular,
\begin{align*}
 \mathbb{P}\left( \left. \overline{W}_n^{(i,r)} > t \,\right|\, X=x \right) &= \mathbb{P}\Big(  W_n^{(i)} - T_r^{(i)} > t \,\Big|\, W_n^{(i)}> T_r^{(i)}, \, X=x \Big),
\end{align*}
for all $t\geq 0$. Using the same argument used to obtain the hazard rate of $W_n^{(i)}$ we get that, conditioned on $X=x$, the hazard rate of $\overline{W}_n^{(i,r)}$ at $t$ is
\begin{equation*}
  h_{\overline{W}_n^{(i,r)}}(t \mid x) = \frac{\mu}{x}\sum\limits_{\ell=1}^n \mathbb{P}\left( \left. T^{(i)}_\ell < t + T_r^{(i)} \,\right|\,W_n^{(i)}> T_r^{(i)}+ t,\, X=x \right),
\end{equation*}
It is easily checked that
\[ h_{\overline{W}_n^{(i,r)}}(t \mid x) \geq h_{W_n^{(i)}}(t \mid x) \]
for all $t\geq 0$, and thus
\begin{equation}
  \mathbb{E}\left[ \left. W_n^{(i)} \right| X=x \right] \geq \mathbb{E}\left[\left. \overline{W}_n^{(i,r)} \right| X=x \right], \label{eq:replicaVsJob}
\end{equation}
for all $r\geq 1$.}

  On the other hand, since the $r$-th replica of phase $i$ starts its service only if $W_n^{(i)}> T_r^{(i)}$, it follows that the probability that a typical job with tasks of size $x$ has at least $r$ replicas associated with it in service during its $i$-th phase is
  \[ \mathbb{P}\left(\left.W_n^{(i)}> T_r^{(i)} \right| X=x \right). \]
  Then, the expected number of replicas associated with a typical job with tasks of size $x$ in phase $i$ is
  \[ \sum\limits_{k=1}^n \mathbb{P}\left(\left.W_n^{(i)}> T_k^{(i)} \right| X=x \right). \]
  This means that the fraction of replicas of size $x$ in phase $i$ that are the $r$-th replica to start phase $i$ in a typical job is
  \[ \frac{ \mathbb{P}\left(\left.W_n^{(i)}> T_r^{(i)}\right| X=x\right)}{\sum\limits_{k=1}^n \mathbb{P}\left(\left.W_n^{(i)}> T_k^{(i)}\right| X=x\right) }. \]
  Therefore, the expected time that a typical replica in phase $i$ is in service is equal to the weighted average
  \begin{align*}
   \mathbb{E} \left[\left. \overline{W}_n^{(i)}\right| X=x \right] &= \sum\limits_{r=1}^n \left[ \frac{ \mathbb{P}\left(\left.W_n^{(i)}> T_r^{(i)}\right| X=x\right)}{\sum\limits_{k=1}^n \mathbb{P}\left(\left.W_n^{(i)}> T_k^{(i)}\right| X=x\right) } \right] \mathbb{E}\left[\left. \overline{W}_n^{(i,r)}\right| X=x \right].
  \end{align*}
  Combining this with Equation \eqref{eq:replicaVsJob}, we obtain
  \begin{align*}
   \mathbb{E} \left[ \left.\overline{W}_n^{(i)}\right| X=x \right] &\leq  \sum\limits_{r=1}^n \left[ \frac{ \mathbb{P}\left(\left.W_n^{(i)}> T_r^{(i)}\right| X=x\right)}{\sum\limits_{k=1}^n \mathbb{P}\left(\left.W_n^{(i)}> T_k^{(i)}\right| X=x\right) } \right] \mathbb{E}\left[ \left.W_n^{(i)}\right| X=x \right] \\
   &= \mathbb{E}\left[ \left. W_n^{(i)}\right| X=x \right],
  \end{align*}
  which concludes the proof of the lemma.
\end{proof}

Combining lemmas \ref{lem:first} and \ref{lem:second}, we obtain
\begin{align*}
 \mathbb{E}\big[W_n \big| X=x \big] &= \sum\limits_{i=0}^{k_n} \mathbb{E}\left[\left.W_n^{(i)}\right|X=x\right] \\
 &\geq x + \sum\limits_{i=1}^{k_n} \mathbb{E}\left[\left.\overline{W}_n^{(i)}\right|X=x\right]\\
 &\geq x\left( 1+\sum\limits_{i=1}^{k_n} \frac{1}{k_n\mu \left( \frac{\tilde{p}(x)}{\lambda p(x) x}- \frac{1}{\mu} \right) - \mu(i-1)} \right),
\end{align*}
for all $x\in\mathcal{X}\backslash\{0\}$. Moreover, combining this lower bound with the fact that the service time of any job with tasks of size $0$ is $0$, i.e. $\mathbb{E}[W_n(0)]=0$, and with the fact that the expected size of tasks is $\eta$, we obtain
\begin{align}
   \mathbb{E}[W_n] &= \sum_{x\in\mathcal{X}} p(x) \mathbb{E}\big[W_n\big| X=x\big] \nonumber \\
   &\geq \sum_{x\in\mathcal{X}\backslash\{0\}} x p(x) \left( 1+\sum\limits_{i=1}^{k_n} \frac{1}{k_n\mu \left( \frac{\tilde{p}(x)}{\lambda p(x) x}- \frac{1}{\mu} \right) - \mu(i-1)} \right) \nonumber \\
   &= \eta + \sum_{x\in\mathcal{X}\backslash\{0\}} \sum\limits_{i=1}^{k_n} \frac{xp(x)}{k_n\mu \left( \frac{\tilde{p}(x)}{\lambda p(x) x}- \frac{1}{\mu} \right) - \mu(i-1)}. \label{eq:unknownBound}
\end{align}

At this point, we have a lower bound for the expected service time of a job, but it depends on the unknown quantities $\{\tilde{p}(x):x\in\mathcal{X}\backslash \{0\}\}$. In the following lemma, we obtain a lower bound that only depends on the system parameters $\lambda$, $\mu$, $\eta$, and $k_n$, by minimizing the lower bound in Equation \eqref{eq:unknownBound} with respect to the unknown quantities, over an appropriate domain.

\begin{lemma}\label{lem:finiteBound}
  We have
  \[ \mathbb{E}[W_n] \geq \eta\left( 1+\frac{1}{\mu}\sum\limits_{i=1}^{k_n} \frac{1}{k_n \left( \frac{1}{\lambda \eta}- \frac{1}{\mu} \right) - i+1} \right). \]
\end{lemma}
\begin{proof}
First, note that Equation \eqref{eq:unknownBound} implies that
\[  \mathbb{E}[W_n] \geq \min\limits_{\tilde{p}} \left\{ \eta + \sum_{x\in\mathcal{X}\backslash\{0\}} \sum\limits_{i=1}^{k_n} \frac{xp(x)}{k_n\mu \left( \frac{\tilde{p}(x)}{\lambda p(x) x}- \frac{1}{\mu} \right) - \mu(i-1)} \right\}, \]
where the minimization is over all ``possible'' quantities $\tilde{p}=\{\tilde{p}(x):x\in\mathcal{X}\backslash\{0\}\}$. These largely depend on the dispatching/replication policy, but they are inherently constrained, as follows.

Recall that $\tilde{p}(x)$ is the expected number of servers (normalized by $n$) working on replicas of size $x$. Since there can only be up to $n$ servers busy in expectation, we have
\begin{equation}\label{eq:constraints1}
 \sum_{x\in\mathcal{X}\backslash\{0\}} \tilde{p}(x) \leq 1.
\end{equation}
Furthermore, since all jobs eventually leave the system, all jobs must have $k_n$ associated replicas that actually finish their service. Thus, jobs with tasks of size $x$ have at least $k_n$ replicas associated with them that are in service at least $x$ units of time. This implies that the arrival rate of replicas of size $x$ that will be in service for at least $x$ units of time has to be greater than or equal to $k_n$ times the arrival rate of jobs with tasks of size $x$, i.e., we must have $\lambda_n^{(0)}(x)\geq \lambda p(x) n$. Combining this with Equation~\eqref{eq:lambda0}, we obtain
\begin{equation}\label{eq:constraints2}
    \tilde{p}(x) \geq \lambda p(x) x\left( 1+\frac{1}{\mu} \right),
\end{equation}
for all $x\in\mathcal{X}\backslash\{0\}$. Using equations \eqref{eq:constraints1} and \eqref{eq:constraints2} to define the domain of our optimization problem in the variables $\{\tilde{p}(x):x\in\mathcal{X}\backslash \{0\}\}$, we get that $\mathbb{E}[W_n]$ is lower bounded by
\begin{align*}
   \min\limits_{\tilde{p}\in[0,1]^{|\mathcal{X}|-1}} \qquad &\eta + \sum_{x\in\mathcal{X}\backslash\{0\}} \sum\limits_{i=1}^{k_n} \frac{xp(x)}{k_n\mu \left( \frac{\tilde{p}(x)}{\lambda p(x) x}- \frac{1}{\mu} \right) - \mu(i-1)} \\
    s.t. \qquad & \tilde{p}(x) \geq \lambda p(x) x\left( 1+\frac{1}{\mu} \right),\qquad \forall x\in\mathcal{X}\backslash\{0\}, \\
    &\sum_{x\in\mathcal{X}\backslash\{0\}} \tilde{p}(x) \leq 1.
\end{align*}
Note that this is a finite-dimensional convex optimization problem. 
Taking the dual problem, and using the necessary condition of stability given in Theorem \ref{thm:stability}, it is easily checked that $\tilde{p}(x)= x p(x)/\eta$ minimizes the objective function. This results in the lower bound
  \begin{align*}
   \mathbb{E}[W_n] &\geq \eta + \sum_{x\in\mathcal{X}} \sum\limits_{i=1}^{k_n} \frac{xp(x)}{k_n\mu \left( \frac{1}{\lambda \eta}- \frac{1}{\mu} \right) - \mu(i-1)} \nonumber \\
   &= \eta\left( 1+\frac{1}{\mu}\sum\limits_{i=1}^{k_n} \frac{1}{k_n \left( \frac{1}{\lambda \eta}- \frac{1}{\mu} \right) - i+1} \right).
  \end{align*}
\end{proof}

\subsection{General task sizes}
In this subsection we consider the case where the task size $X$ has a general distribution with unit mean. Let us fix an admissible policy for the infinite-server relaxation. For each positive integer $m$, we consider a coupled system, as follows.
\begin{itemize}
  \item [a)] The arrival process is the same.
  \item [b)] If a job with tasks of size $X$ arrives to the original system, then a job with tasks of size
\[ X^{(m)} \triangleq \min \left\{2^m,\,\, \frac{\left\lceil X 2^m \right\rceil -1}{2^m} \right\} \leq X \]
arrives to the $m$-th coupled system.
  \item [c)] For each job, replicas are always dispatched to idle servers, with the same relative phase start times as in the original system (unless the phase ends earlier due to the shrunk task sizes).
  \item [d)] The slowdowns are the same as in the original system.
\end{itemize}
Since the replicas have the same relative phase start times (unless the phase ends earlier) and are subject to the same slowdowns, and since $X^{(m)}\leq X$, we have
\[ \min\limits_{r=1,\dots,n} \Big\{ T_r^{(i)}(X)+X^{(m)}S_r^{(i)} \Big\} \leq \min\limits_{r=1,\dots,n} \Big\{ T_r^{(i)}(X)+XS_r^{(i)} \Big\}, \]
for all $i=1,\dots,k_n$, and for all $m$. In particular, this implies that the policy in the $m$-th coupled system has at most $n$ servers busy in expectation, and that the expected delay of a job in the $m$-th coupled system is smaller than or equal to the one in the original system, i.e., that we have
\begin{equation}\label{eq:delayComparisson}
  \mathbb{E}\big[W_n\big] \geq \mathbb{E}\left[W_n^{(m)}\right],
\end{equation}
for all $m$.

Finally, note that the replicas in the $m$-th coupled system can only take values in the finite set $\big\{ l/2^m : 0\leq l\leq 2^{2m} \big\}$. Then, we can use Lemma \ref{lem:finiteBound} to obtain the lower bound
\[ \mathbb{E}\left[W_n^{(m)}\right]\geq \mathbb{E}[X_m]\left( 1+\frac{1}{\mu}\sum\limits_{i=1}^{k_n} \frac{1}{k_n \left( \frac{1}{\lambda \mathbb{E}[X_m]}- \frac{1}{\mu} \right) - i+1} \right), \]
which holds for all $m\geq 1$. Combining this with Equation \eqref{eq:delayComparisson}, we obtain
\begin{equation}
  \mathbb{E}\big[W_n\big]\geq \sup_{m\geq 1} \left\{\mathbb{E}[X_m]\left( 1+\frac{1}{\mu}\sum\limits_{i=1}^{k_n} \frac{1}{k_n \left( \frac{1}{\lambda \mathbb{E}[X_m]}- \frac{1}{\mu} \right) - i+1} \right) \right\}. \label{eq:supBound}
\end{equation}
We now proceed to compute this supremum. Note that
\[ \mathbb{P}(X_m > x)=\mathbb{P}\left(X> \frac{\left\lfloor x 2^m \right\rfloor + 1}{2^m}\right), \]
for all $x< 2^m$, and $\mathbb{P}(X_m> x)=0$, for all $x\geq 2^m$. Using the fact that the sequence of events $\big\{X>\big(\left\lfloor x 2^m \right\rfloor +1\big) 2^{-m} \big\}_{m\geq 1}$ is nondecreasing for any given $x\geq 0$, we have
\begin{align*}
    \lim_{m\to\infty} \mathbb{P}(X_m\geq x) &= \lim_{m\to\infty} \mathbb{P}\left(X> \frac{\left\lfloor x 2^m \right\rfloor + 1}{2^m}\right) \\
    &= \mathbb{P}\left(\bigcup_{m=1}^\infty \left\{X > \frac{\left\lfloor x 2^m \right\rfloor + 1}{2^m}\right\}\right) \\
    &= \mathbb{P}\left(X > x\right),
\end{align*}
for all $x\geq 0$. Thus, the monotone convergence theorem implies
\begin{equation*}
    \lim_{m\to\infty} \mathbb{E}[X_m] = \lim_{m\to\infty} \int\limits_{0}^{\infty} \mathbb{P}(X_m> x) dx = \int\limits_{0}^{\infty} \mathbb{P}(X> x) dx = \mathbb{E}[X] = 1.
\end{equation*}
Combining this with Equation \eqref{eq:supBound}, we obtain
\[  \mathbb{E}[W_n]\geq  1+\frac{1}{\mu}\sum\limits_{i=1}^{k_n} \frac{1}{k_n \left( \frac{1}{\lambda}- \frac{1}{\mu} \right) - i+1}, \]
which concludes the proof of Theorem \ref{thm:boundGeneral}.

\section{Proof of Theorem 3.4}\label{proof:upperBound}
We will be working under the assumption that $k_n=1$. In this case, the FREC policy has the following convenient property.

\begin{lemma}\label{lem:startFinish}
 Under the FREC policy for $k_n=1$, all replicas associated with the same job start their service at the same time and leave the system at the same time.
\end{lemma}
\begin{proof}
Since $k_n=1$, it is immediate that all replicas associated with the same job leave the system as the same time, that is, as soon as some replica completes service. Furthermore, it is also immediate that all replicas associated with jobs that find enough idle servers start their service at the same time.

We now focus on the replicas associated with jobs that are routed to the first subsystem, but do not find enough idle servers. Recall that when there are less than $\lceil r^* \rceil$ idle servers, replicas are created and dispatched to \emph{all} servers in the subsystem, and they are not cancelled until $\lceil r^* \rceil$ replicas have started their service (or until one that did start its service, finishes). As a result, there can be at most one job with less than $\lceil r^* \rceil$ associated replicas in service. Combining this with the fact that $n^{(1)}$ is multiple of $\lceil r^* \rceil$, we conclude that all jobs have either $\lceil r^* \rceil$ or zero associated replicas in service, and that the number of busy servers is always a multiple of $\lceil r^* \rceil$. Thus, replicas that were dispatched to all servers start their service when the $\lceil r^* \rceil$ replicas associated with another job leave the system. Thus, all $\lceil r^* \rceil$ replicas associated with the same job start their service at the same time. The same argument applies for the second subsystem.
\end{proof}

Lemma \ref{lem:startFinish} implies that under the FREC policy for $k_n=1$, we can think of each set of $\lceil r^* \rceil$ or $\lceil r^* \rceil-1$ replicas associated with the same job as a single job being processed by a single server. Combined with the fact that queued replicas under the FREC policy wait for the first set of $\lceil r^* \rceil$ or $\lceil r^* \rceil-1$ servers to become idle in a first-come first-serve fashion, it follows that the first and second subsystems behave as M/G/$m_1$ and M/G/$m_2$ queues, respectively, with
\[ m_1\triangleq \frac{n^{(1)}}{\lceil r^* \rceil} = \left\lfloor \lambda n\left( p_1-2n^{\alpha-1} \right)^+\left(1+\frac{1}{\mu\lceil r^* \rceil} \right) + \frac{\lambda n^{\alpha}}{\lceil r^* \rceil} \right\rfloor, \]
and
\[ m_2\triangleq \frac{n^{(2)}}{\lceil r^* \rceil-1} =\left\lfloor \lambda n\left[1-\left( p_1-2n^{\alpha-1} \right)^+\right]\left(1+\frac{1}{\mu(\lceil r^* \rceil-1)} \right) + \frac{\lambda n^{\alpha}}{\lceil r^* \rceil -1} \right\rfloor. \]
In particular, since the arrival processes to both subsystems are obtained as independent thinnings of the original Poisson process of arrivals, they are Poisson processes as well, with rates
\[ \lambda_n^{(1)} \triangleq \lambda n (p_1-2n^{\alpha-1})^+, \]
and
\[ \lambda_n^{(2)} \triangleq \lambda n \left[1-\left(p_1-2n^{\alpha-1}\right)^+\right]. \]
Moreover, since all the replicas associated with the same job start and finish their service at the same time (cf. Lemma \ref{lem:startFinish}), the service times of the jobs in the first and second subsystems are i.i.d. and distributed as $X(1+\min\{S_1,\dots,S_{\lceil r^* \rceil}\})$ and $X(1+\min\{S_1,\dots,S_{\lceil r^* \rceil-1}\})$, respectively. Thus, the expected service times are
\[ s^{(1)}\triangleq 1+\frac{1}{\mu\lceil r^* \rceil}, \qquad \text{and} \qquad s^{(2)}\triangleq 1+\frac{1}{\mu(\lceil r^* \rceil-1)}. \]
Hence, the loads in the first and second subsystems are
\begin{align*}
   \rho_n^{(1)} &\triangleq \frac{\lambda_n^{(1)}s^{(1)}}{m_1} =\frac{\lambda n (p_1-2n^{\alpha-1})^+\left( 1+\frac{1}{\mu\lceil r^* \rceil} \right)}{\left\lfloor \lambda n\left( p_1-2n^{\alpha-1} \right)^+\left(1+\frac{1}{\mu\lceil r^* \rceil} \right) + \frac{\lambda n^{\alpha}}{\lceil r^* \rceil} \right\rfloor}<1,
\end{align*}
and
\begin{align*}
   \rho_n^{(2)} &\triangleq \frac{\lambda_n^{(2)}s^{(2)}}{m_2 } = \frac{\lambda n \left[1-\left(p_1-2n^{\alpha-1}\right)^+\right]\left( 1+\frac{1}{\mu(\lceil r^* \rceil-1)} \right)}{\left\lfloor \lambda n\left[1-\left( p_1-2n^{\alpha-1} \right)^+\right]\left(1+\frac{1}{\mu(\lceil r^* \rceil-1)} \right) + \frac{\lambda n^{\alpha}}{\lceil r^* \rceil -1} \right\rfloor}<1,
\end{align*}
respectively. As a result, we have
\begin{align*}
   1-\rho_n^{(1)} &\approx \frac{\frac{\lambda n^{\alpha}}{\lceil r^* \rceil}}{\left\lfloor \lambda n\left( p_1-2n^{\alpha-1} \right)^+\left(1+\frac{1}{\mu\lceil r^* \rceil} \right) + \frac{\lambda n^{\alpha}}{\lceil r^* \rceil} \right\rfloor} \in \Theta\left( n^{\alpha-1} \right),
\end{align*}
and
\begin{align*}
   1-\rho_n^{(2)} &\approx \frac{\frac{\lambda n^{\alpha}}{\lceil r^* \rceil -1}}{\left\lfloor \lambda n\left[1-\left( p_1-2n^{\alpha-1} \right)^+\right]\left(1+\frac{1}{\mu(\lceil r^* \rceil-1)} \right) + \frac{\lambda n^{\alpha}}{\lceil r^* \rceil -1} \right\rfloor} \in \Omega\left( n^{\alpha-1} \right),
\end{align*}

Since the two subsystems behave exactly as M/G/$m_1$ and M/G/$m_2$ queues, their positive Harris recurrence is given by Theorem 2.2 and Corollary 2.8 in \cite{asmussen}. This guarantees the stability of the FREC policy.\\

For the service time, we have
\begin{align*}
   \lim_{n\to\infty} \mathbb{E}[W_n^s] &= \lim_{n\to\infty} \left[ 1+\frac{\left(p_1-2n^{\alpha-1}\right)^+}{\mu\lceil r^* \rceil} + \frac{1-\left(p_1-2n^{\alpha-1}\right)^+}{\mu(\lceil r^* \rceil-1)} \right] \\
   &= 1+\frac{p_1}{\mu\lceil r^* \rceil} + \frac{1-p_1}{\mu(\lceil r^* \rceil-1)}.
\end{align*}
It only remains to show that
\[ \lim_{n\to\infty} \mathbb{E}[W_n^q] = 0. \]
Since both subsystems behave as M/G/m queues, and since there exists $\epsilon>0$ such that $\mathbb{E}[X^{2+\epsilon}]<\infty$, Corollary 2 in \cite{goldberg17} states that there exists a constant $C_\epsilon$, independent from $n$, such that
\begin{align*}
 \mathbb{E}[W^q_n] &\leq C_\epsilon \left[\frac{\left(p_1-2n^{\alpha-1}\right)^+}{\lambda_n^{(1)}\left(1-\rho^{(1)}_n\right)} + \frac{1-\left(p_1-2n^{\alpha-1}\right)^+}{\lambda_n^{(2)}\left(1-\rho^{(2)}_n\right)} \right] \\
 &= C_\epsilon \left[\frac{1}{\lambda n\left(1-\rho^{(1)}_n\right)} + \frac{1}{\lambda n\left(1-\rho^{(2)}_n\right)} \right].
\end{align*}
Combining this with the fact that $1-\rho_n^{(1)} \in \Theta\left( n^{\alpha-1} \right)$ and $1-\rho_n^{(2)} \in \Omega\left( n^{\alpha-1} \right)$, we obtain
\begin{align*}
  \lim_{n\to\infty} \mathbb{E}[W^q_n] &\leq \lim_{n\to\infty} C_\epsilon \left[\frac{1}{\lambda n\left(1-\rho^{(1)}_n\right)} + \frac{1}{\lambda n\left(1-\rho^{(2)}_n\right)} \right] = 0,
\end{align*}
which concludes the proof.

\section{Proof of Theorem 3.5}\label{proof:upperBound2}

The proof is done in two steps. First, we show that the queueing delay converges to zero in probability (Subsection \ref{sec:vanishingDelay}), and then we show the convergence of the expected service time of a typical job to the desired constant (Subsection \ref{sec:convergenceServiceTime}).

\subsection{Vanishing queueing delay}\label{sec:vanishingDelay}
In this subsection we show that the large server pools defined by the DQ policy (i.e., the first $n^{(1)}-\lambda n^\alpha/2k_n$ and $n^{(2)}-\lambda n^\alpha/2k_n$ servers of the first and second subsystems, respectively) are stable, and that the queueing delay of a typical job converges to zero. We provide the complete proof of these facts for the first subsystem, with the second one being analogous.\\

For $i=1,\dots,k_n$, let $Q^{(1)}_i(\cdot)$ be the process describing the number of jobs in the large server pool of the first subsystem, for which less than $i$ replicas associated with it have finished their service. Since all replicas associated with jobs sent to the large server pool start their service at the same time, $Q^{(1)}_i(\cdot)$ also describes the number of replicas in the large server pool of the first subsystem whose service time distributed as $X (1+S_{[i:\lceil r^* k_n \rceil]})$, where $S_{[i:\lceil r^* k_n \rceil]}$ is the $i$-th order statistic of $\lceil r^* k_n \rceil$ slowdowns.


Recall that replicas arrive to the first subsystem in batches, as a Poisson process of rate
\[ \lambda_n^{(1)} \triangleq \frac{\lambda n \big(p_{k_n}-2n^{\alpha-1}\big)^+}{k_n} \in O\left( \frac{n}{k_n} \right), \]
and they are either dispatched to idle servers in the large server pool, or diverted to the small server pool (i.e., the last $\lambda n^\alpha/2k_n$ servers). As a result, for $i=1,\dots,k_n$, Proposition 2 in \cite{stochasticDominance} implies that $Q^{(1)}_i(\cdot)$ is dominated by the queue length process, $\tilde{Q}^{(i)}_i(\cdot)$, of a queueing system where all the arrivals to the first subsystem are sent to idle servers. This corresponds to an M/G/$\infty$ queue with Poison arrivals of rate $\lambda_n^{(1)}$, and i.i.d. jobs distributed as $X (1+S_{[i:\lceil r^* k_n \rceil]})$, with expected service time
\[ s^{(1)}_i \triangleq 1+\mathbb{E}\left[ S_{[i:\lceil r^* k_n \rceil]} \right] \in \Theta(1). \]
Moreover, for $i=1,\dots,k_n$, Proposition 2 in \cite{stochasticDominance} implies that $\tilde{Q}^{(1)}_i(\cdot)$ is dominated by the queue length process, $\overline{Q}^{(1)}_i(\cdot)$, of an M/G/$n^{(1)}_i$ queue with
\begin{equation}
 n^{(1)}_i \triangleq \left\lfloor \frac{\lambda n \big(p_{k_n}-2n^{\alpha-1}\big)^+}{k_n}\Big( 1+\mathbb{E}\left[ S_{[i:\lceil r^* k_n \rceil]} \right] \Big) + \frac{\lambda n^\alpha}{2k_n \lceil r^* k_n \rceil} -1 \right\rfloor, \label{eq:MGm_sizes}
\end{equation}
and the same arrivals and job sizes. It follows that, for $i=1,\dots,k_n$, we have
\begin{equation}
 Q^{(1)}_i(t) \leq \overline{Q}^{(1)}_i(t), \label{eq:processesDominance}
\end{equation}
for all $t\geq 0$. Since the load of the $i$-th M/G/$n^{(1)}_i$ queue is
\begin{align}
 \rho^{(1)}_{n,i} &\triangleq \frac{\lambda_n^{(1)} s^{(1)}_i}{n^{(1)}_i} = \frac{\frac{\lambda n\big(p_{k_n}-2n^\alpha\big)^+}{k_n} \left( 1+\mathbb{E}\left[ S_{[i:\lceil r^* k_n \rceil]} \right] \right)}{ \left\lfloor \frac{\lambda n\big(p_{k_n}-2n^\alpha\big)^+}{k_n}\left( 1+\mathbb{E}\left[ S_{\left[i:\lceil r^* k_n \rceil\right]} \right] \right) + \frac{\lambda n^\alpha}{2k_n \lceil r^* k_n \rceil} -1 \right\rfloor}<1, \label{eq:queueLoad}
\end{align}
then it is Harris recurrent (cf. \cite{asmussen}), for $i=1,\dots,k_n$. Combining this with Equation \eqref{eq:processesDominance}, we conclude that the original queues are also Harris recurrent. We denote the steady-state queue length processes of the original and of the larger queues by $Q^{(1)}_1,\dots,Q^{(1)}_{k_n}$ and $\overline{Q}^{(1)}_1,\dots,\overline{Q}^{(1)}_{k_n}$, respectively. Given the stochastic ordering of the original queue length processes, we have
\begin{equation}\label{eq:queue_dominance}
  Q^{(1)}_i \leq_{st} \overline{Q}^{(1)}_i,
\end{equation}
for all $i=1,\dots,k_n$.\\

Let $W_n^{q,1}$ be the queueing delay of a typical job in the first subsystem. The vanishing of this delay is established in the following lemma.

\begin{lemma}
We have
\[ \lim\limits_{n\to\infty} \mathbb{P}\left(W_n^{q,1}>0\right) = 0. \]
\end{lemma}
\begin{proof}
Let $Q^{(1)}$ be the steady-state number of replicas present in the large server pool of the first subsystem. Since there is no queueing in the large server pool, $Q^{(1)}$ is also the steady-state number of busy servers in the large server pool of the first subsystem. Recall that under the DQ policy, a job can have positive queueing delay only if its diverted to the small server pool, which happens when there are less than $\lceil r^* k_n \rceil$ idle servers (or equivalently, more than $n^{(1)}- \lambda n^\alpha/(2k_n) -\lceil r^* k_n \rceil$ busy servers) in the large server pool at the time of the job's arrival. Since the arrivals are Poisson, the PASTA property implies that the steady-state probability of waiting is
\begin{align}
 \mathbb{P}\left(W_n^{q,1}>0\right) &\leq \mathbb{P}\left(Q^{(1)} > n^{(1)}- \frac{\lambda n^\alpha}{2k_n} -\lceil r^* k_n \rceil \right). \label{eq:queueingDelay}
\end{align}
Here there is an inequality because jobs sent to the small server pool might still experience zero queueing delay. Moreover, using the definition of the integers $n^{(1)}_1,\dots,n^{(1)}_{k_k}$ given in Equation \eqref{eq:MGm_sizes}, and the fact that, for $i=1,\dots,k_n$, $Q^{(1)}_i$ is the steady-state number of jobs in the large server pool of the first subsystem, for which less than $i$ replicas associated with it have finished their service, it can be checked that
\[ n^{(1)}_1+\cdots+n^{(1)}_{k_n-1}+\Big(\lceil r^* k_n \rceil-k_n+1\Big)n^{(1)}_{k_n}\leq n^{(1)}- \frac{\lambda n^\alpha}{2k_n}-\lceil r^* k_n \rceil, \]
and that
\[ Q^{(1)} = \left(\sum\limits_{i=1}^{k_n-1} Q^{(1)}_i \right) + \Big(\lceil r^* k_n \rceil-k_n+1 \Big)Q^{(1)}_{k_n}. \]
Combining these two facts with Equation \eqref{eq:queue_dominance}, it follows that
\begin{align}
 &\mathbb{P}\left(Q^{(1)} > n^{(1)}- \frac{\lambda n^\alpha}{2k_n} -\lceil r^* k_n \rceil \right) \nonumber \\
 &\qquad = \mathbb{P}\left(\left(\sum\limits_{i=1}^{k_n-1} Q^{(1)}_i \right) + \big(\lceil r^* k_n \rceil-k_n+1 \big)Q^{(1)}_{k_n} > n^{(1)}- \frac{\lambda n^\alpha}{2k_n} -\lceil r^* k_n \rceil \right) \nonumber \\
 &\qquad\leq \mathbb{P}\left(\left(\sum\limits_{i=1}^{k_n-1} Q^{(1)}_i \right) + \big(\lceil r^* k_n \rceil-k_n+1 \big)Q^{(1)}_{k_n} > n^{(1)}_1+\cdots+n^{(1)}_{k_n-1} \right. \nonumber \\
 & \qquad\qquad\qquad\qquad\qquad\qquad\qquad\qquad\qquad\qquad\quad \left. +\Big(\lceil r^* k_n \rceil-k_n+1\Big)n^{(1)}_{k_n} \right) \nonumber \\
 &\qquad\leq \mathbb{P}\left(\bigcup_{i=1}^{k_n} \left\{Q^{(1)}_i > n^{(1)}_i\right\} \right) \nonumber \\
 &\qquad\leq \sum\limits_{i=1}^{k_n} \mathbb{P}\left( Q^{(1)}_i > n^{(1)}_i \right) \nonumber \\
 &\qquad\leq \sum\limits_{i=1}^{k_n} \mathbb{P}\left( \overline{Q}^{(1)}_i > n^{(1)}_i \right). \label{eq:delayBound}
\end{align}
Recall that, for $i=1,\dots,k_n$, $\overline{Q}^{(1)}_i$ is the steady-state queue length of an M/G/$n_i^{(1)}$ queue with load $\rho^{(1)}_{n,i}$ (cf. Equation \eqref{eq:queueLoad}), such that
\[ 1- \rho^{(1)}_{n,i} \approx \frac{\frac{\lambda n^\alpha}{2k_n \lceil r^* k_n \rceil}}{ \left\lfloor \frac{\lambda n\big(p_{k_n}-2n^\alpha\big)^+}{k_n}\left( 1+\mathbb{E}\left[ S_{\left[i:\lceil r^* k_n \rceil\right]} \right] \right) + \frac{\lambda n^\alpha}{2k_n \lceil r^* k_n \rceil} -1 \right\rfloor} \in \Theta\left( \frac{n^{\alpha-1}}{k_n} \right). \]
Since there exists $\epsilon>0$ such that $\mathbb{E}[X^{2+\epsilon}]<\infty$, Corollary 1 in \cite{goldberg17} states that there exists a constant $C_\epsilon$, independent from $n$ and $i$, such that
\[ \mathbb{P}\left( \overline{Q}^{(1)}_i > n^{(1)}_i \right) \leq \frac{C_\epsilon}{n^{(1)}_i\left(1-\rho^{(1)}_{n,i}\right)^2}. \]
Combining this with the fact that $1- \rho^{(1)}_{n,i} \in \Theta\left( n^{\alpha-1}/k_n \right)$, we obtain
\begin{equation*}
 \sum\limits_{i=1}^{k_n} \mathbb{P}\left( \overline{Q}^{(1)}_i > n^{(1)}_i \right) \leq \sum\limits_{i=1}^{k_n} \frac{C_\epsilon}{n^{(1)}_i\left(1-\rho^{(1)}_{n,i}\right)^2} \in O\left( \frac{k_n^4}{n^{2\alpha-1}} \right).
\end{equation*}
Combining this with equations \eqref{eq:queueingDelay} and \eqref{eq:delayBound}, we get that
\begin{equation}
 \mathbb{P}\left(W_n^{q,1}>0\right) \leq \mathbb{P}\left(Q^{(1)} > n^{(1)}- \frac{\lambda n^\alpha}{2k_n} -\lceil r^* k_n \rceil \right) \in O\left( \frac{k_n^4}{n^{2\alpha-1}} \right). \label{eq:queueingOrder}
\end{equation}
Finally, since we assumed that $k_n\in O(n^\beta)$ and that $\alpha>(4\beta+2)/3$, then
\[ \lim\limits_{n\to\infty} \mathbb{P}\left(W_n^{q,1}>0\right) = 0. \]
\end{proof}

Similarly, if $W_n^{q,2}$ is the queueing delay of a typical job in the second subsystem, it can be shown that
\[ \lim\limits_{n\to\infty} \mathbb{P}\left(W_n^{q,2}>0\right) = 0. \]
Finally, using the fact that a job is sent to the first subsystem with probability $\big( p_{k_n} - 2n^{\alpha-1} \big)^+$ and to the second subsystem with probability $1-\big( p_{k_n} - 2n^{\alpha-1} \big)^+$, we get that the probability of a typical job having a positive queueing delay is
\begin{align*}
 \lim\limits_{n\to\infty} \mathbb{P}\left(W_n^{q}>0\right) &= \lim\limits_{n\to\infty} \left( p_{k_n}-2n^{\alpha-1} \right)^+\mathbb{P}\left(W_n^{q,1}>0\right) \\
 &\qquad\qquad\qquad\qquad + \left[ 1-\left( p_{k_n}-2n^{\alpha-1} \right)^+  \right] \mathbb{P}\left(W_n^{q,2}>0\right) \\
 &= 0.
\end{align*}

\subsection{Convergence of the expected service time}\label{sec:convergenceServiceTime}
In this subsection we show that the small server pools defined by the DQ policy are stable, and that the expected service time of a typical job converges to the desired constant. We provide the complete proof of these facts for the first subsystem, with the second one being analogous.\\

Recall that the arrival rate to the first subsystem is $\lambda n \big( p_{k_n} - 2n^{\alpha-1} \big)^+/k_n$, and that jobs are only sent to the small server pool if there are less than $\lceil r^* k_n \rceil$ idle servers in the large server pool, which happens with probability
\begin{equation*}
 p_{small}^{(1)}\triangleq \mathbb{P}\left(Q^{(1)} > n^{(1)}- \frac{\lambda n^\alpha}{2k_n} -\lceil r^* k_n \rceil \right)
\end{equation*}
in steady-state. Then, the arrival rate of replicas to the small server pool is
\begin{align*}
 \lambda^{(1)}_{small} \triangleq \frac{\lambda n \big( p_{k_n} - 2n^{\alpha-1} \big)^+ p_{small}^{(1)}}{k_n}.
\end{align*}
Combining this with Equation \eqref{eq:queueingOrder}, we have
\begin{align}
 \lambda^{(1)}_{small} \in O\left( \frac{k_n^3}{n^{2\alpha-2}} \right). \label{eq:arrivalRate}
\end{align}
Since only $k_n$ replicas are created for each job sent to the small server pool, no replicas are prematurely cancelled, and thus the expected service time of each replica is
\[ s^{(1)}\triangleq 1+\frac{1}{\mu}. \]
Combining this with Equation \eqref{eq:arrivalRate} and the fact that the small server pool has $\lambda n^\alpha/2k_n$ servers, we get that the load of the small server pool of the first subsystem is
\[ \rho^{(1)}_{small} \triangleq \frac{\lambda^{(1)}_{small} s^{(1)}}{\frac{\lambda n^\alpha}{2k_n}} \in O\left( \frac{k_n^4}{n^{3\alpha-2}} \right). \]
Since we assumed that $k_n\in O(n^\beta)$ and $\alpha>(4\beta+2)/3$, then the load of the small server pool converges to zero, and it is thus stable for all $n$ large enough.\\

Let $\mathbb{E}\big[W_n^{s,1}\big]$ be the expected service time of a typical job in the first subsystem. We have the following convergence results.

\begin{lemma}
If $k_n=k$ for all $n$, we have
\begin{align*}
 \lim\limits_{n\to\infty} \mathbb{E}\left[W_n^{s,1}\right] &= 1+ \frac{1}{\mu}\sum\limits_{i=1}^{k} \frac{1}{\lceil r^* k \rceil -i+1}.
\end{align*}
If $k_n\to\infty$ as $n\to\infty$, we have
\begin{align*}
 \lim\limits_{n\to\infty} \mathbb{E}\left[W_n^{s,1}\right] &= 1+ \frac{1}{\mu}\log\left( \frac{r^*}{r^*-1} \right)
\end{align*}
\end{lemma}
\begin{proof}
Recall that under the DQ policy, all replicas associated with the same job sent to the large server pools start their service at the same time by construction. As a result, the service time of a job sent to the large server pool in first subsystem is distributed as the $k_n$-th order statistic of $X(1+S_1),\dots,X(1+S_{\lceil r^* k_n \rceil})$. Thus, its expected service time is
\[ 1+ \frac{1}{\mu}\sum\limits_{i=1}^{k_n} \frac{1}{\lceil r^* k_n \rceil -i+1}. \]

On the other hand, note that the service time of any job is upper bounded by the sum of the service times of its replicas, by the definition of the service time of a job. Since jobs sent to the small server pool have only $k_n$ replicas associated with them, then the expected service time of a job sent to the small server pool in the first subsystem, to be denoted by $s_{small}^{(1)}$, is upper bounded as follows:
\begin{equation}
  s_{small}^{(1)} \leq k_n\left( 1+\frac{1}{\mu} \right). \label{eq:serviceBound}
\end{equation}
Since a job is sent to the small server pool within the first subsystem with probability $p_{small}^{(1)}$, we see that the expected service time of a typical job in the first subsystem is
\[ \mathbb{E}\left[W_n^{s,1}\right] \triangleq 1+ p_{small}^{(1)}s_{small}^{(1)} + \left( 1-p_{small}^{(1)} \right) \frac{1}{\mu}\sum\limits_{i=1}^{k_n}  \frac{1}{\lceil r^* k_n \rceil -i+1}. \]
Combining this with equations \eqref{eq:queueingOrder} and \eqref{eq:serviceBound}, and the facts that $k_n\in O(n^\beta)$ and $\alpha>(5\beta+1)/2$, we obtain
\begin{align*}
 \lim\limits_{n\to\infty} \mathbb{E}\left[W_n^{s,1}\right] &= 1+ \frac{1}{\mu}\sum\limits_{i=1}^{k} \frac{1}{\lceil r^* k \rceil -i+1}
\end{align*}
for the case where $k_n=k$ for all $n$, and
\begin{align*}
 \lim\limits_{n\to\infty} \mathbb{E}\left[W_n^{s,1}\right] &= 1+ \frac{1}{\mu}\log\left( \frac{r^*}{r^*-1} \right)
\end{align*}
for the case where $k_n\to\infty$ as $n\to\infty$.
\end{proof}

Similarly, if $\mathbb{E}\big[W_n^{s,2}\big]$ is the expected service time of a typical job in the second subsystem, it can be shown that
\begin{align*}
 \lim\limits_{n\to\infty} \mathbb{E}\left[W_n^{s,2}\right] &= 1+ \frac{1}{\mu}\sum\limits_{i=1}^{k} \frac{1}{\lceil r^* k \rceil -i}
\end{align*}
for the case where $k_n=k$ for all $n$, and
\begin{align*}
 \lim\limits_{n\to\infty} \mathbb{E}\left[W_n^{s.,2}\right] &= 1+ \frac{1}{\mu}\log\left( \frac{r^*}{r^*-1} \right)
\end{align*}
for the case where $k_n\to\infty$ as $n\to\infty$.

Finally, using the fact that a job is sent to the first subsystem with probability $\big( p_{k_n} - 2n^{\alpha-1} \big)^+$ and to the second subsystem with probability $1-\big( p_{k_n} - 2n^{\alpha-1} \big)^+$, we get that the expected service time of a typical job in the system is
\[ \mathbb{E}\left[ W^s_n \right] = \big( p_{k_n} - 2n^{\alpha-1} \big)^+ \mathbb{E}\left[W_n^{s,1}\right] + \left[1-\big( p_{k_n} - 2n^{\alpha-1} \big)^+\right] \mathbb{E}\left[W_n^{s,2}\right],  \]
and thus
\begin{align*}
 \lim\limits_{n\to\infty} \mathbb{E}\left[W_n^s\right] &= 1+ \frac{1}{\mu}\sum\limits_{i=1}^{k} \left( \frac{p_k}{\lceil r^* k \rceil -i+1} + \frac{1-p_k}{\lceil r^* k \rceil -i} \right)
\end{align*}
for the case where $k_n=k$ for all $n$, and
\begin{align*}
 \lim\limits_{n\to\infty} \mathbb{E}\left[W_n^s\right] &= 1+ \frac{1}{\mu}\log\left( \frac{r^*}{r^*-1} \right)
\end{align*}
for the case where $k_n\to\infty$ as $n\to\infty$.

\section{Proof of Lemma 4.1}\label{app:problemFormulation}

For Block policies, all replicas associated with the same job start their service at the same time. As a result, the expected service time of job with tasks of size $x$ for which exactly $r$ replicas start their service is
\[ 1+\mathbb{E}\big[ S_{[k_n:r]}\,\big|\, X=x \big]. \]
Moreover, the expected total server time that each job requires in this case is
\[ (r-k_n)x\Big(1+\mathbb{E}\big[S_{[k_n:r]}\,\big|\, X=x \big]\Big) + \sum\limits_{i=1}^{k_n} x\Big(1+\mathbb{E}\big[ S_{[i:r]}\,\big|\, X=x \big] \Big). \]
Indeed, the second term is the sum of the smallest $k_n$ service times, which correspond to the service times of replicas that do finish their service. The first term is the total server time of the $(r-k_n)$ replicas that do not finish their service, but are nevertheless in service for $x\big(1+\mathbb{E}\big[S_{[k_n:r]}\,\big|\, X=x \big]\big)$ units of time.

Averaging these expressions over the number of replicas created (according to the distribution $p(x)$) and over the possible task sizes (according to the task size distribution $\mathbb{P}_X$), the expected service time of a typical job is
\begin{equation*}
 \int\limits_0^\infty x \sum\limits_{r=k_n}^\infty p_r(x)\Big(1+\mathbb{E}\big[ S_{[k_n:r]}\,\big|\, X=x \big] \Big) d\mathbb{P}_X(x),
\end{equation*}
and the expected server time that a typical job requires is
\begin{align}
 &\int\limits_0^\infty x \sum\limits_{r=k_n}^\infty p_r(x) \left[ r + (r-k_n)\mathbb{E}\big[S_{[k_n:r]}\,\big|\, X=x \big] {\color{white}\sum\limits_{i=1}^{k_n}} \right. \nonumber \\
  &\qquad\qquad\qquad\qquad\qquad\qquad\qquad + \left. \sum\limits_{i=1}^{k_n} \mathbb{E}\big[ S_{[i:r]}\,\big|\, X=x \big] \right] d\mathbb{P}_X(x). \label{eq:above}
\end{align}
On the other hand, by Little's law, the expected number of busy servers in steady-state is equal to the arrival rate of jobs ($\lambda n/k_n$) times the expected server time (in Equation \eqref{eq:above}). Since this must be less than or equal to the total number of servers ($n$), we have
\begin{align*}
 &\frac{\lambda}{k_n} \int\limits_0^\infty x \sum\limits_{r=k_n}^\infty p_r(x) \left[ r + (r-k_n)\mathbb{E}\big[S_{[k_n:r]}\,\big|\, X=x \big] {\color{white}\sum\limits_{i=1}^{k_n}} \right. \\
  &\qquad\qquad\qquad\qquad\qquad\qquad\qquad\qquad + \left. \sum\limits_{i=1}^{k_n} \mathbb{E}\big[ S_{[i:r]}\,\big|\, X=x \big] \right] d\mathbb{P}_X(x) \leq 1.
\end{align*}

\section{Proof of Lemma 4.2}\label{app:feasibility}
Note that the condition
  \[ \lambda\leq\frac{1}{1+\frac{1}{\mu}}, \]
  implies that the problem is feasible ($p_{k_n}(x)=1$ for all $x$ is a feasible solution). Moreover, the objective function is lower bounded by $1$. These two facts imply that the infimum is finite.

  Let us now define the function $I:\mathbb{R}\to\mathbb{R}$ as the optimal objective value of the perturbed problem:
  \begin{align*}
    I(u)\triangleq \inf_{p\in \mathcal{P}_{k_n}} \quad &\int\limits_0^\infty x \sum\limits_{r=k_n}^\infty p_r(x)\Big(1+\mathbb{E}\big[ S_{[k_n:r]}\,\big|\, X=x \big] \Big) d\mathbb{P}_X(x) \\
    s.t.\quad & \frac{\lambda}{k_n} \int\limits_0^\infty x \sum\limits_{r=k_n}^\infty p_r(x) \left[ r + (r-k_n)\mathbb{E}\big[S_{[k_n:r]}\,\big|\, X=x \big] \right. \\
    &\qquad\qquad\qquad\qquad \left. + \sum\limits_{i=1}^{k_n} \mathbb{E}\big[ S_{[i:r]}\,\big|\, X=x \big] \right] d\mathbb{P}_X(x) \leq 1+u.
  \end{align*}
  Moreover, consider the function $g:\mathcal{P}_{k_n}\to\mathbb{R}$, defined by
  \begin{align*}
    g(p) &\triangleq \frac{\lambda}{k_n} \int\limits_0^\infty x \sum\limits_{r=k_n}^\infty p_r(x) \left[ r + (r-k_n)\mathbb{E}\big[S_{[k_n:r]}\,\big|\, X=x \big] {\color{white}\sum\limits_{i=1}^{k_n}} \right.  \\
     &\qquad\qquad\qquad\qquad\qquad\qquad\qquad \left. + \sum\limits_{i=1}^{k_n} \mathbb{E}\big[ S_{[i:r]}\,\big|\, X=x \big] \right] d\mathbb{P}_X(x) - 1.
  \end{align*}
  It is easy to check that
  \[ 0 \in \text{core}\big[g(\mathcal{P}_{k_n})+\mathbb{R}_+\big], \]
  which implies that $I(\cdot)$ is relatively continuous at $0$, and thus its subdifferential at zero is non-empty. Combining this with the fact that the infimum is finite, we use Theorem 4 in \cite{infLinProg} to conclude that the infimum is attained.

\section{Proof of Theorem 4.3}\label{proof:sizeBasedK}

We start with a technical result about the optimization problem.

\begin{lemma}\label{lem:strongDuality}
  Consider the function $I:\mathbb{R}\to\mathbb{R}$ defined as the optimal objective value of the perturbed problem:
  \begin{align*}
    I(u)\triangleq \inf_{p\in \mathcal{P}_{k_n}} \quad &\int\limits_0^\infty x \sum\limits_{r=k_n}^\infty p_r(x)\Big(1+\mathbb{E}\big[ S_{[k_n:r]}\,\big|\, X=x \big] \Big) d\mathbb{P}_X(x) \\
    s.t.\quad & \frac{\lambda}{k_n} \int\limits_0^\infty x \sum\limits_{r=k_n}^\infty p_r(x) \left[ r + (r-k_n)\mathbb{E}\big[S_{[k_n:r]}\,\big|\, X=x \big] \right. \\
    &\qquad\qquad\qquad\qquad \left. + \sum\limits_{i=1}^{k_n} \mathbb{E}\big[ S_{[i:r]}\,\big|\, X=x \big] \right] d\mathbb{P}_X(x) \leq 1+u.
  \end{align*}
  Then we have that $I(\cdot)$ is relative continuous at $0$, and we have the strong duality
  \begin{align}
    I(0)= \sup\limits_{y\geq 0} \inf\limits_{p\in \mathcal{P}_{k_n}} \quad &\int\limits_0^\infty x \sum\limits_{r=1}^{\infty} p_r(x)\left[ 1-y+ \frac{y\lambda r}{k_n} + \mathbb{E}\big[S_{[k_n:r]}\,\big|\, X=x\big] \right. \nonumber \\
    &\quad \left( 1+\frac{y\lambda}{k_n}(r-k_n) \right) \left. + \frac{y\lambda}{k_n}  \sum\limits_{i=1}^{k_n} \mathbb{E}\big[ S_{[i:r]}\,\big|\, X=x \big] \right] d\mathbb{P}_X(x). \label{eq:dual_problem}
  \end{align}
  Moreover, $I(0)$ is attained by the dual problem at points $(p^*,y^*)$ with $y^*>0$.
\end{lemma}
\begin{proof}
  The relative continuity of $I(\cdot)$, the strong duality, and the attainability of the optimal value by the dual are obtained using the same arguments as in the proof of Lemma \ref{lem:feasibility}, given in Appendix \ref{app:feasibility}.

  Moreover, since the inequality constraint in the problem is clearly active, i.e., since we have $0\notin\partial I(0)$, then $y^*>0$ for every optimal dual solution $(p^*,y^*)$.
\end{proof}

Since we have strong duality (Lemma \ref{lem:strongDuality}), we focus on the solution of the dual problem:
  \begin{align*}
    \sup\limits_{y\geq 0} \inf\limits_{p\in \mathcal{P}_{k_n}} \quad &\int\limits_0^\infty x \sum\limits_{r=1}^{\infty} p_r(x)\left[ 1-y+ \frac{y\lambda r}{k_n} + \mathbb{E}\big[S_{[k_n:r]}(x)\big] \left( 1+\frac{y\lambda}{k_n}(r-k_n) \right) \right. \\
    &\qquad\qquad\qquad\qquad\qquad\qquad\qquad \left. + \frac{y\lambda}{k_n}  \sum\limits_{i=1}^{k_n} \mathbb{E}\big[ S_{[i:r]}(x) \big] \right] d\mathbb{P}_X(x).
  \end{align*}
  Furthermore, since Lemma \ref{lem:strongDuality} states that the supremum with respect to $y$ is attained at points with $y>0$, we can take the supremum over all $y>0$ and obtain the same solutions. Moreover, since there are no constraints involving multiple values of $x$ at the same time, we can take the infimum for each $x$ separately, inside the integral. Using these two facts, we obtain the equivalent problem:
  \begin{align*}
    \sup\limits_{y> 0}  \quad &\int\limits_0^\infty  x \inf\limits_{p(x)\in P_{k_n}} \left\{ \sum\limits_{r=1}^{\infty} p_r(x)\left[ 1-y+ \frac{y\lambda r}{k_n} + \mathbb{E}\big[S_{[k_n:r]}\,\big|\, X=x\big] \right.\right. \\
    &\qquad\qquad\quad\left( 1+\frac{y\lambda}{k_n}(r-k_n) \right) \left.\left.  + \frac{y\lambda}{k_n}  \sum\limits_{i=1}^{k_n} \mathbb{E}\big[ S_{[i:r]}\,\big|\, X=x \big] \right] \right\} d\mathbb{P}_X(x).
  \end{align*}

  We now explore the properties of the solutions. Fix $x\geq 0$ and $y>0$, and consider the problem
  \begin{align*}
   \inf\limits_{p(x)\in P_{k_n}} & \left\{ \sum\limits_{r=1}^{\infty} p_r(x)\left[ 1-y+ \frac{y\lambda r}{k_n} + \mathbb{E}\big[S_{[k_n:r]}\,\big|\, X=x\big] \left( 1+\frac{y\lambda}{k_n}(r-k_n) \right)\right.\right. \\
   & \qquad\qquad\qquad\qquad\qquad\qquad\qquad\qquad \left.\left.+ \frac{y\lambda}{k_n}  \sum\limits_{i=1}^{k_n} \mathbb{E}\big[ S_{[i:r]}\,\big|\, X=x \big] \right] \right\}.
  \end{align*}
  Since this minimization is a linear program over an infinite-dimensional simplex, any optimal solution, if it exists, will be concentrated on the set of indices
  \[ \underset{r\geq k_n}{\arg\min} \left\{ \frac{y\lambda r}{k_n} + \mathbb{E}\big[S_{[k_n:r]}\,\big|\, X=x\big] \left( 1+\frac{y\lambda}{k_n}(r-k_n) \right) + \frac{y\lambda}{k_n}  \sum\limits_{i=1}^{k_n} \mathbb{E}\big[ S_{[i:r]}\,\big|\, X=x \big] \right\}. \]
  We now proceed to explore how these indices depend on $x$ and $k_n$. Let us define the function $f_{x,y}:\mathbb{N}\to\mathbb{R}$ such that
  \[ f_{x,y}(r)\triangleq \frac{y\lambda r}{k_n} + \mathbb{E}\big[S_{[k_n:r]}\,\big|\, X=x\big] \left( 1+\frac{y\lambda}{k_n}(r-k_n) \right) + \frac{y\lambda}{k_n}  \sum\limits_{i=1}^{k_n} \mathbb{E}\big[ S_{[i:r]}\,\big|\, X=x \big]. 
  \]
  For $r\geq 2$, we have
  \begin{align*}
    &f_{x,y}(r)-f_{x,y}(r-1) \\
    &\qquad = \frac{y\lambda}{k_n}\left( 1 + \sum\limits_{i=1}^{k_n} \mathbb{E}\big[ S_{[i:r]}\,\big|\, X=x \big] - \mathbb{E}\big[ S_{[i:r-1]}\,\big|\, X=x \big] \right) \\
    &\qquad\qquad + \Big( \mathbb{E}\big[ S_{[k_n:r]}\,\big|\, X=x \big] - \mathbb{E}\big[ S_{[k_n:r-1]}\,\big|\, X=x \big] \Big) \left[ 1+\frac{y\lambda}{k_n}(r-k_n) \right] \\
    &\qquad\qquad\qquad\qquad\qquad\qquad\qquad\qquad\qquad\qquad\quad + \frac{y\lambda}{k_n}\mathbb{E}\big[ S_{[k_n:r-1]}\,\big|\, X=x \big].
  \end{align*}
  Since
  \[ \lim\limits_{x\to\infty} \Big( \mathbb{E}\big[ S_{[i:r]}\,\big|\, X=x \big]-\mathbb{E}\big[ S_{[i:r-1]}\,\big|\, X=x\big] \Big) = 0 \]
  and
  \[ \lim\limits_{x\to\infty} \mathbb{E}\big[ S_{[k_n:r-1]}\,\big|\, X=x \big] = \frac{1}{\mu}, \]
  then $f_{x,y}(r)-f_{x,y}(r-1)>0$, for all $x$ large enough. Thus, $r^*(x,y)=k_n$ for all $x$ large enough. This proves part (i) of the theorem.\\

  On the other hand, we have
  \begin{align*}
     f_{x,y}(r) &\leq \frac{y\lambda r}{k_n} + \mathbb{E}\big[S_{[k_n:r]}\,\big|\, X=x\big] \left( 1+\frac{y\lambda}{k_n}(r-k_n) \right) \\
     & \qquad\qquad\qquad\qquad\qquad\qquad\qquad\qquad\qquad + \frac{y\lambda}{k_n}  \sum\limits_{i=1}^{k_n} \mathbb{E}\big[ S_{[k_n:r]}\,\big|\, X=x \big]\\
     &= \frac{y\lambda r}{k_n} + \mathbb{E}\big[S_{[k_n:r]}\,\big|\, X=x\big] \left( 1+\frac{y\lambda r}{k_n} \right).
  \end{align*}
  Since $\mathbb{E}\big[S_{[k_n:r]}\,\big|\, X=x\big]-1/\mu$ is either increasing or decreasing in $x$, we have that either
  \[ \mathbb{E}\big[ S_{[k_n:r]}\,\big|\, X=x \big] \leq \mathbb{E}\big[S_{[k_n:r]}\mid X=0\big], \]
  or
  \[ \mathbb{E}\big[ S_{[k_n:r]}\,\big|\, X=x \big] \leq \lim\limits_{x\to\infty} \mathbb{E}\big[S_{[k_n:r]}\,\big|\, X=x\big] = \frac{1}{\mu}. \]
  Either way, we have
  \[ \mathbb{E}\big[ S_{[k_n:r]}\,\big|\, X=x \big] \leq \max\left\{\mathbb{E}\big[S_{[k_n:r]}\mid X=0\big], \frac{1}{\mu} \right\}, \]
  for all $r$. Moreover, since $S_{[k_n:r]}$ is the $k_n$-th order statistic of $r$ i.i.d. random variables with mean $1/\mu$, there exists $\underline{r}\in O(k_n)$ such that $\mathbb{E}\big[S_{[k_n:r]}\mid X=0\big]\leq 1/\mu$ for all $r\geq \underline{r}$. Thus,
  \[ f_{x,y}(r) \leq \frac{y\lambda r}{k_n} + \frac{1}{\mu} \left( 1+\frac{y\lambda r}{k_n} \right) \triangleq \overline{f_{x,y}}(r), \]
  for all $r\geq \underline{r}$. Combining this with the fact that
  \[ f_{x,y}(r) \geq \frac{y\lambda r}{k_n}\triangleq \underline{f_{x,y}}(r), \]
  we see that $f_{x,y}$ is sandwiched between the two affine functions $\underline{f_{x,y}}$ and $\overline{f_{x,y}}$, for all $r\geq \underline{r}$. It follows that the minimum of $f_{x,y}$ is achieved at indices upper bounded by the largest positive integer $\overline{r^*_{k_n}}$ such that
  \[ \overline{f_{x,y}}\big(\underline{r}\big) \geq \underline{f_{x,y}}\big(\overline{r^*_{k_n}}\big). \]
  Equivalently, for all $y>0$, we have
  \[ \overline{r^*_{k_n}} \leq \underline{r} + \frac{1}{\mu} \left( \frac{k_n}{y\lambda} +\underline{r} \right) \in O(k_n). \]
  Finally, since this holds for all $y>0$, in particular it holds for the optimal $y^*>0$, and part (ii) of the theorem is proved.

\section{Proof of Theorem 4.5}\label{proof:sizeBasedThm}
We first present a simple result on the expectation of the minimum of i.i.d. random variables.

\begin{lemma}\label{lem:convexity}
  Let $S_1,S_2,\dots$ be a sequence of nonnegative and non constant i.i.d. random variables. Then,
  \[ g(r)\triangleq \mathbb{E}\big[\min\{S_1,\dots,S_r\}\big] \]
  is a strictly convex function, i.e., we have
  \[ g(r) < \frac{g(r-1)+g(r+1)}{2}, \]
  for all $r\geq 2$.
\end{lemma}
\begin{proof}
We have
  \begin{align*}
   &\mathbb{E}[\min\{S_1,\dots,S_{r+1}\}] \\
   &\qquad = \mathbb{E}\big[\min\{S_1,\dots,S_r\}\big] - \mathbb{E}\Big[\big(\min\{S_1,\dots,S_r\}-S_{r+1}\big) \mathds{1}_{\{\min\{S_1,\dots,S_r\}>S_{r+1}\}}\Big].
  \end{align*}
  This means that
  \begin{align*}
   g(r+1)-g(r) &= - \mathbb{E}\Big[\big(\min\{S_1,\dots,S_r\}-S_{r+1}\big) \mathds{1}_{\{\min\{S_1,\dots,S_r\}>S_{r+1}\}} \Big] \\
   & = - \mathbb{E}\Big[ \mathbb{E}\Big[ \big(\min\{S_1,\dots,S_r\}-S_{r+1}\big) \mathds{1}_{\{\min\{S_1,\dots,S_r\}>S_{r+1}\}} \mid S_{r+1} \Big] \Big].
  \end{align*}
  Since $S_1,\dots,S_{r+1}$ are independent, we have
  \begin{align*}
   &\mathbb{E}\Big[ \mathbb{E}\Big[ \big(\min\{S_1,\dots,S_r\}-S_{r+1}\big) \mathds{1}_{\{\min\{S_1,\dots,S_r\}>S_{r+1}\}} \mid S_{r+1} \Big] \Big] \\
   &\qquad\qquad\qquad = \int\limits_0^\infty \mathbb{E}\Big[ \big(\min\{S_1,\dots,S_r\}-s\big) \mathds{1}_{\{\min\{S_1,\dots,S_r\}>s\}} \Big] dF_{S_{r+1}}(s).
  \end{align*}
  Note that, for all $s\geq 0$, the integrand is a nonincreasing function of $r$. Moreover, since $S_{r+1}$ is not constant, the integrand is a decreasing function of $r$ for a set of values of $s$ with positive probability (with respect to $F_{S_{r+1}}$). It follows that $g(r+1)-g(r)$ is a decreasing function of $r$. As a result,
  \[ [g(r+1)-g(r)] - [g(r)-g(r-1)] > 0, \]
  and thus
  \[ g(r)< \frac{g(r-1)+g(r+1)}{2}. \]
\end{proof}

For the case $k_n=1$, the optimization problem defined by equations \eqref{eq:problem1} and \eqref{eq:problem2} simplifies to the following.
\begin{align*}
    \inf\limits_{p\in\mathcal{P}_1}\quad &\int\limits_0^\infty x \sum\limits_{r=1}^{\infty} p_r(x)\Big(1+\mathbb{E}\big[ S_{[1:r]}\,\big|\, X=x \big]\Big) d\mathbb{P}_X(x) \\
    s.t.\quad & \lambda \int\limits_0^\infty x \sum\limits_{r=1}^{\infty} r p_r(x)\Big(1+\mathbb{E}\big[ S_{[1:r]}\,\big|\, X=x \big]\Big) d\mathbb{P}_X(x) \leq 1.
\end{align*}
Since we have strong duality (Lemma \ref{lem:strongDuality}), we focus on the solution of the dual problem:
  \begin{align*}
    \sup\limits_{y\geq 0} \inf\limits_{p\in \mathcal{P}_1} \quad &\int\limits_0^\infty x \sum\limits_{r=1}^{\infty} p_r(x)\Big[ \Big( 1+\mathbb{E}\big[ S_{[1:r]}\,\big|\, X=x \big]\Big)(1+y\lambda r)-y \Big] d\mathbb{P}_X(x)
  \end{align*}
  Furthermore, since Lemma \ref{lem:strongDuality} states that the supremum with respect to $y$ is attained at points with $y>0$, we can take the supremum over all $y>0$ and obtain the same solutions. Moreover, since there are no constraints involving multiple values of $x$ at the same time, we can take the infimum for each $x$ separately, inside the integral. Using these two facts, we obtain the equivalent problem:
  \begin{align*}
    \sup\limits_{y> 0}  \quad &\int\limits_0^\infty  x \inf\limits_{p(x)\in P_1} \left\{ \sum\limits_{r=1}^{\infty} p_r(x)\Big[ \Big( 1+\mathbb{E}\big[ S_{[1:r]}\,\big|\, X=x \big]\Big)(1+y\lambda r)-y \Big] \right\} d\mathbb{P}_X(x).
  \end{align*}

  We now explore the properties of the solutions. Fix $x\geq 0$ and $y>0$, and consider the problem
  \[ \inf\limits_{p(x)\in P_1}  \quad \sum\limits_{r=1}^{\infty} p_r(x)\Big[ \Big( 1+\mathbb{E}\big[ S_{[1:r]}\,\big|\, X=x \big]\Big)(1+y\lambda r)-y \Big]. \]
  Since this is a linear program over a simplex, the solutions are concentrated on the set of indices
  \[ \underset{r\geq 1}{\arg\min} \Big[ \Big(1+\mathbb{E}\big[ S_{[1:r]}\,\big|\, X=x \big]\Big)(1+y\lambda r)-y \Big]. \]
  We now proceed to explore how these indices depend on $x$. Let us define the function $f_{x,y}:\mathbb{N}\to\mathbb{R}$ by
  \[ f_{x,y}(r)\triangleq \Big(1+\mathbb{E}\big[ S_{[1:r]}\,\big|\, X=x \big]\Big)(1+y\lambda r)-y. \]
  First, note that Lemma \ref{lem:convexity} states that $\mathbb{E}\big[ S_{[1:r]}\,\big|\, X=x \big]$ is a strictly convex function in $r$. Furthermore, Assumption \ref{ass:convexity} states that $r\mathbb{E}\big[ S_{[1:r]}\,\big|\, X=x \big]$ is also a convex function in $r$. As a result, $f_{x,y}(r)$ is a strictly convex function in $r$. Moreover, since $y>0$, we have
  \[ \lim\limits_{r\to\infty} f_{x,y}(r)=\infty. \]
  Combining this with the fact that $f_{x,y}$ is strictly convex, we can conclude that the minimum of $f_{x,y}$ is achieved in either one point, or in two consecutive points. This proves part (i) of the theorem.\\

  On the other hand, for $r\geq 2$, consider
  \begin{align*}
    f_{x,y}(r)-f_{x,y}(r-1) &= (1+y\lambda r)\Big( \mathbb{E}\big[ S_{[1:r]}\,\big|\, X=x \big]-\mathbb{E}\big[ S_{[1:r-1]}\,\big|\, X=x \big]\Big) \\
    &\qquad\qquad\qquad\qquad\qquad + y\lambda\Big(1+\mathbb{E}\big[ S_{[1:r-1]}\,\big|\, X=x \big]\Big).
  \end{align*}
  By our assumptions on the slowdowns, $\mathbb{E}\big[ S_{[1:r-1]}\,\big|\, X=x \big]$ and $\mathbb{E}\big[ S_{[1:r]}\,\big|\, X=x \big]-\mathbb{E}\big[ S_{[1:r-1]}\,\big|\, X=x \big]$ are increasing in $x$, for all $r\geq 2$. As a result, $f_{x,y}(r)-f_{x,y}(r-1)$ is increasing in $x$, for all $r\geq 2$. Since $f_{x,y}$ is convex, it follows that the minimum of $f_{x,y}$ is achieved at integers $r^*(x,y)$, which are nonincreasing with $x$. This proves part (iii) of the theorem.\\


Finally, note that the minimum of $f_{x,y}(\cdot)$ is achieved at two consecutive integers $r^*_{x,y}$ and $r^*_{x,y}+1$ only if $f_{x,y}(r^*_{x,y}+1)-f_{x,y}(r^*_{x,y})=0$. Since this difference is strictly increasing in $x$, the minimum can only be achieved at a certain pair of consecutive integers for only one value of $x$, call it $\hat x$. Thus, if $\mathbb{P}(X=\hat x)=0$, optimal solutions are equal almost everywhere to a solution that is concentrated only on $r^*_{x,y}$. This proves part (ii) of the theorem.


\section{Proof of theorems 4.6 and 4.7} \label{sec:proofSizeBasedUpper}

Note that the SB-FREC and SB-DQ policies are almost the same as the FREC and DQ policies introduced in Subsection \ref{sec:upperBoundFinite}. The main difference is that there are more than two subsystems (but still finitely many of them), and that the routing of incoming jobs to the subsystems depends on the task sizes of the incoming jobs. Thus, while it is clear that the SB-DQ policy is a Block policy by construction, the fact that the SB-FREC policy is a Block policy is established using the same argument that was used to establish that the FREC policy is a Block policy (Lemma~\ref{lem:startFinish}).

Moreover, the convergence of the queueing delays to zero under the SB-FREC and SB-DQ policies follow the same arguments given in the proofs of theorems \ref{thm:UpperBound1} and \ref{thm:UpperBoundK} in sections \ref{proof:upperBound} and \ref{proof:upperBound2}, respectively. It only remains to prove that the expected service times converge to the stated limits.\\

Consider the function $I:\mathbb{R}\to\mathbb{R}$ given by
\begin{align*}
  I(u)= \inf_{p\in \mathcal{P}_{k}} \quad &\int\limits_0^\infty x \sum\limits_{r=k}^\infty p_r(x)\Big(1+\mathbb{E}\big[ S_{[k:r]}\,\big|\, X=x \big] \Big) d\mathbb{P}_X(x) \\
  s.t.\quad & \frac{\lambda}{k} \int\limits_0^\infty x \sum\limits_{r=k}^\infty p_r(x) \left[ r + (r-k)\mathbb{E}\big[S_{[k:r]}\,\big|\, X=x \big] \right. \\
  &\qquad\qquad\qquad\qquad \left. + \sum\limits_{i=1}^{k_n} \mathbb{E}\big[ S_{[i:r]}\,\big|\, X=x \big] \right] d\mathbb{P}_X(x) \leq 1+u.
\end{align*}
By construction, the expected service time of a typical job under the SB-FREC and SB-DQ policies is
\begin{equation*}
 \mathbb{E}[W_n^s] = \int\limits_0^\infty x \sum\limits_{r=k}^\infty p^{(n)}_r(x)\Big(1+\mathbb{E}\big[ S_{[k:r]}\,\big|\, X=x \big] \Big) d\mathbb{P}_X(x),
\end{equation*}
Moreover, the definition of $p^{(n)}$ implies that
\begin{equation}\label{eq:serviceTime}
  \mathbb{E}[W_n^s] = I\big(-n^{\alpha-1}\big).
\end{equation}
On the other hand, the definition of $p^*$ implies that
\begin{equation}\label{eq:limitTime}
   \int\limits_0^\infty x \sum\limits_{r=k}^\infty p^*_r(x)\Big(1+\mathbb{E}\big[ S_{[k:r]}\,\big|\, X=x \big] \Big) d\mathbb{P}_X(x) = I(0).
\end{equation}
Finally, combining equations \eqref{eq:serviceTime} and \eqref{eq:limitTime} with the fact that $I(\cdot)$ is continuous around $0$ (Lemma \ref{lem:strongDuality}), we obtain
\begin{align*}
  \lim\limits_{n\to\infty} \mathbb{E}[W_n^s] &= \lim\limits_{n\to\infty} I\big(-n^{\alpha-1}\big) \\
  &= I(0) \\
  &= \int\limits_0^\infty x \sum\limits_{r=k}^\infty p^*_r(x)\Big(1+\mathbb{E}\big[ S_{[k:r]}\,\big|\, X=x \big] \Big) d\mathbb{P}_X(x),
\end{align*}
which proves the convergence of the expected service times.

\bibliographystyle{imsart-number}
\bibliography{references}

\end{document}